\documentclass[12pt]{article}
\usepackage[utf8]{inputenc}

\usepackage{latexsym,amssymb,amsfonts,amsmath,amsthm}
\usepackage[dvips]{graphicx}
\usepackage{comment}

\usepackage{ulem}

\newtheorem{theorem}{Theorem}
\newtheorem{lemma}{Lemma}

\newtheorem{proposition}{Proposition}
\newtheorem{remark}{Remark}
\newtheorem{definition}{Definition}

\newtheorem{assumption}{Assumption}

\numberwithin{theorem}{section}
\numberwithin{lemma}{section}
\numberwithin{corollary}{section}
\numberwithin{proposition}{section}
\numberwithin{remark}{section}

\newcommand{\bs}[1]{\boldsymbol{#1}}

\newcommand{\supp}{{\rm supp}}

\newcommand{\dist}{{\rm dist}}

\usepackage{color}

\title{Design for implementation of discrete-time  quantum walk with circulant matrix on graph by optical polarizing elements}

\author{
Yusuke Mizutani$^1$, Etsuo Segawa$^{2}$, \\
Yusuke Higuchi$^{3}$, Leo Matsuoka$^{4}$, Tomoyuki Horikiri$^{1}$  \\
{\small $^1$ Faculty of Engineering Division of Intelligent Systems Engineering, Yokohama National University,}\\
{\small Hodogaya, Yokohama 240-8501, Japan.} \\
{\small $^2$Graduate School of Environment and Information Sciences, Yokohama National University,}\\
{\small Hodogaya, Yokohama 240-8501, Japan.
}\\
{\small $^3$Department of Mathematics,
Gakushuin University, } \\
{\small Tokyo 171-8588, Japan.} \\
{\small $^4$Faculty of Engineering, Hiroshima Institute of Technology, } \\
{\small Hiroshima, 731-5193, Japan.} \\
}

\date{}

\begin{document}

\maketitle

\par\noindent
{\bf Abstract}. 
In this paper, we introduce a quantum walk whose local scattering at each vertex is denoted by a unitary circulant matrix; namely the circulant quantum walk. We also introduce another quantum walk induced by the circulant quantum walk; namely the optical quantum walk, whose underlying graph is a $2$-regular directed graph and obtained by blowing up the original graph in some way. We propose a design of an optical circuit which implements the stationary state of the optical quantum walk. We show that if the induced optical quantum walk does not have $+1$ eigenvalue, then the stationary state of the optical quantum walk gives that of the original circulant quantum walk. From this result, we give a useful condition for the setting of the circulant quantum walks which can be implemented by this optical circuit. 

\section{Introduction}
Random walks on finite graphs play key roles to analysis on the electrical network~(e.g., \cite{DS}) and cut off phenomena~(e.g., \cite{Diaconis,LevinPeres}). 
Quantum walks (QWs) are known as quantum counterpart of such random walks~\cite{AmbainisEtAl,FH}. 
Although the application of quantum walks to quantum search algorithm is one of the more remarkable application~\cite{Ambainis2003, Childs, Portugal},  
we also anticipate that such a quantum version of the application will be developed. 
It is well known that irreducible random walks on finite graph converges to the stationary state in the long time limit. 
This is very fundamental to such applications of random walks.
The convergence to the stationary state of random walks  is supported by the fact that all the absolute values of the eigenvalues except $\pm 1$ are strictly smaller than $1$ because these eigenvalues converge to $0$ by the exponentiation of the time steps.   
On the other hand, the eigenvalues of quantum walks live on the unit circle in the complex plan. Thus every eigenvalue having the overlap to the initial state ``asserts" its existence even in the long time limit in general since the  absolute value of the eigenvalue is unit. 
Indeed, in the quantum search algorithm, 
a high probability at the marked vertices is obtained in an asymptotic periodicity of the time evolution with respect to the sufficiently large number of vertices.  
This derives from the eigenvalues having a large overlap to the initial state.   
Then if we ``miss" the optical timing of the observation, we may have a very low probability of finding the marked vertex~\cite{Brassard}.  

Thus it is natural to consider finding a stationary state of a quantum walk as a fixed point of a dynamical system~\cite{Grover,YLC}. In~\cite{FelHil1,FelHil2,HS}, such a quantum walk model where the dynamics converges to a stationary state is proposed by considering a semi-infinite system. 
In this model, the boundary of the graph and the initial state are set so that the unitary time evolution on the whole space, which includes the outside of the graph, describes that some quantum walkers penetrate as the inflow from the outside of the graph to the boundaries and  some quantum walkers go out from the boundaries to the outside of the graph as the outflow at every time step. 
Then it is mathematically shown the dynamical system based on this quantum walk model converges to a fixed point as a stationary state since the inflow to the graph and outflow to the outside are balanced in the long time limit~\cite{FelHil1,FelHil2,HS}. 
For example, it is shown that the stationary state of the Szegedy walk induced by reversible random walk with a constant inflow can be expressed by using the current of an electrical circuit~\cite{MHS}. 

Then, because discrete-time quantum walk is  implemented on the one- and two-dimensional lattices (see for examples, \cite{ZhaoEtAl,SchreiberEtAl}, \cite{MW} and its references therein), and also continuous-time quantum walk is implemented on the circulant graph~\cite{QLMAZOWM} and so on, we attempt to consider an implementation of this quantum walk on a general graph that has the stationary state as a dynamical system. 
In~\cite{MMOS}, the stationary state of a discrete-time quantum walk model on the one-dimensional lattice gives the stationary Schr{\"o}dinger equation with delta potentials on $\mathbb{R}$~\cite{Albe} 
and a possible approach to the implementation of this quantum walk model using the optical circuit is suggested; the internal graph  corresponds to a finite path graph in the setting of our quantum walk treated here. The more detailed mathematical discussion on \cite{MMOS} from the view point of the spectral and scattering theory can be seen in ~\cite{Morioka}.   
In this paper, based on an idea inspired by \cite{ZhaoEtAl,MMOS} in particular, we propose the design of an optical circuit implementing our quantum walk model on a general graph which converges to a stationary state. 
Moreover we mathematically find a useful setting of this quantum walk model where the stationary states can be implemented by using this optical circuit, although technical difficulties remain in terms of the implementation. 
The quantum walk model of our target for the implementation is called the circulant quantum walk. The quantum coin assigned at each vertex, which describes the manner of scattering at each vertex $u\in V$, is given by a $d(u)$-dimensional circulant matrix. Here $d(u)$ is the degree of the vertex $u$.  
The circulant matrix is diagonalized by the discrete Fourier matrix and related to coding theories ~\cite{Davis}. The circulant matrix assigned at each vertex  coincides with the scattering matrix describing the response of our proposed optical circuit (see Fig.~\ref{fig:a}). 

The main idea of the implementation of the quantum coin is that we embed this local optical circuit into each vertex as the quantum coin and we regard each boundary of the circuit as the ``gateway" leading into one of the adjacent vertices (see Figs.~\ref{fig:GOGBU}, \ref{fig:GprimeO}).
We heuristically show that this resulting large optical circuit can be represented by the stationary state of ``another" quantum walk model-namely, an optical quantum walk.  The underlying graph of the optical quantum walk is a blow-up digraph of the original graph which is $2$-regular. This $2$-regularity is represented by horizontal and vertical polarization of the light  $|H\rangle$ and $|V\rangle$ in the optical circuit in our proposed very ideal design. 
To implement the original quantum walk by this optical circuit, the theoretical problem is to clarify the relation between  stationary states of the circulant quantum walk and optical quantum walk. 
So if such a stationary state of the circulant quantum walk coincides with that of the optical quantum walk, we say that ``the optical quantum walk implements the underlying circulant quantum walk". 
In this paper, we mathematically show a useful sufficient condition for the implementation.  
The sufficient condition provides a concrete setting of the implementing optical circuit. The setting breaks a kind of symmetricity (see Theorem~\ref{cor:graph} and Figs.~\ref{fig:setting1} and \ref{fig:setting2}) with respect to the orientation of the circuit or with respect to the quantum coins. 

The rest of this paper is organized as follows.
In Section~2, the circulant quantum walk on graph $G=(V,A)$ is introduced. The quantum coin assigned at each vertex is a circulant matrix induced by a two-dimensional unitary matrix $H_u$. The circulant quantum walk is determined by the sequence $\{H_u\}_{u\in V}$ and labeling $\{\xi_u\}_{u\in V}$, where $\xi_u: A_u \to \{0,\dots,\deg(u)-1\}$ is a bijection map. Here $A_u$ is the set of all the arcs of $G$ whose terminal vertices are $u\in V$. In Section~3, we introduce the optical quantum walk on the blow-up directed graph induced by the original graph of the circulant graph. Then the optical quantum walk is determined by the same parameters of the circulant quantum walk. Our target is to find a useful condition for the setting of the circulant quantum walk in which the stationary state of the circulant quantum walk can be obtained by referring to that of the optical quantum walk. The motivation for the target derives from the expectation that the optical quantum walk can be implemented by an optical circuit using polarizing elements proposed by Section~4.  
In particular, the design of the optical circuit is proposed for the circulant quantum walk on arbitrary connected graph in Section~4. 
In Section~5, we demonstrate numerically the case for the complete graph with $10$ vertices, $K_{10}$. The first example is the case that the implementation works while the second example is the case that the implementation does not work.  In Section~6, a solution to the concrete labeling way $\{\xi_u\}_{u\in V}$ for the implementation is proposed using a key proposition. 
The key proposition gives a sufficient condition for the implementation: if the induced optical quantum walk does not have eigenvalue $1$, then the implementation works. 
In Section~7, we present the useful conditions for the implementation in the complete graph $K_N$ case.
In Section~8, we give the proof for the key proposition. Finally, provide the summary and discussion of our results.  
The important notations are listed in Table~\ref{table:notation}.

\begin{table}[htb]\label{table:notation}
\caption{Notations in this paper}
\begin{tabular}{r|l}
Symmetric digraph
& a digraph where every arc has the inverse arc\\

$t(a)$, $o(a)$ 
& terminal and origin vertices of arc $a$, respectively \\

$\bar{a}$
& the inverse arc of arc $a$ \\

$G_0=(V_0,A_0)$     
& internal graph (symmetric digraph)  \\

tail
& semi-infinite length path whose root connects to a  vertex of $G_0$ \\

$\tilde{G}_0=(\tilde{V}_0,\tilde{A}_0)$     
& semi-infinite graph obtained by adding the tail to every vertex of $G_0$\\

$\deg(u)$
& degree of $u\in \tilde{V}_0$ in $\tilde{G}_0$ \\

$\partial\tilde{A}_{\pm}$
& the set of arcs of tails \\
& \quad whose origin($+$) and terminus($-$) belong to $G_0$. \\

$A_u$
& the set of arcs of $\tilde{G}_0$ whose terminal vertices are commonly $u\in \tilde{V}$. \\

$\bs{\xi}=(\xi_u)_{u\in \tilde{V}_0}$ 
& labeling of arcs in $\tilde{G}_0$ (Definition~\ref{def:label}) \\

$\tilde{G}_0^{(BU,\xi)}=(\tilde{V}_0^{(BU,\xi)},\tilde{A}_0^{(BU,\xi)})$ 
& blow-up graph of $\tilde{G}_0$ with the labeling $\bs{\xi}$ (Definition~\ref{def:BUG}) \\

island $u$ $(\subset \tilde{G}_0^{(BU,\xi))})$ 
& directed cycle induced by $u\in V_0$ obtained by \\
& \quad the blowing up procedure (1) in Definition~\ref{def:BUG}\\

$A_0^{cycle} \;(\subset \tilde{A}_0^{(BU,\xi)})$
& set of arcs of all the islands \\

$A_{0,u}^{cycle}$
& arc set of the island $u$ \\

$(u;\xi_u(v)) \;(\in \tilde{V}_0^{(BU,\xi)})$ 
& vertex in the island $u$ connected to the island $v$ \\

$G'=(HWP\sqcup PBS,A')$
& the graph for the optical circuit design induced by $\tilde{G}_0^{(BU,\xi)}$ (Sect.~\ref{subsect:optcircuit})\\

&\\

$H_u$ 
& two-dimensional unitary matrix assigned at $u\in V_0$ \\

$\mathrm{Circ}(H_u)$ 
& $\deg(u)$-dimensional circulant matrix defined by (\ref{eq:circulantmat}) \\

$QW(G_0;(H_u)_{u\in V_0};(\xi)_{u\in V_0})$ 
& circulant quantum walk on $\tilde{G}_0$ (Definition~\ref{def:CQW})\\

$\mathrm{Opt}(QW(G_0;(H_u)_{u\in V_0};(\xi)_{u\in V_0}))$
& optical quantum walk on $\tilde{G}_0^{(BU,\xi)}$ (Definition~\ref{def:OQW})\\

$\tilde{U}_0$ 
& time evolution operator of the circulant quantum walk on $\tilde{G}_0$ \\

$U^{(BU)}$
&
time evolution operator of the optical quantum walk on $\tilde{G}_0^{(BU,\xi)}$\\

& \\

$\psi_\infty$
& the stationary state of the circulant quantum walk \\

$\psi_\infty^{(BU)}$
& the stationary state of the optical quantum walk \\

& \\

$\chi_u: \mathbb{C}^{\tilde{A}}\to \mathbb{C}^{[\deg(u)]}$
& restriction to $\mathbb{C}^{\{a\in \tilde{A}_0 \;|\; t(a)=u\}}\cong\mathbb{C}^{[\deg(u)]}$ \\
& \quad following the labeling $\xi$ (Sect.~\ref{subsec:CQW})\\

$\iota: \mathbb{C}^{\tilde{A}_0^{(BU,\xi)}}\to \mathbb{C}^{\tilde{A}_0}$ 
& restriction to $\mathbb{C}^{\tilde{A}_0}$ (Sect.~\ref{set:proof})\\

$\eta_u: \mathbb{C}^{\tilde{A}^{(BU,\xi)}_0}\to \mathbb{C}^{A^{cycle}_{0,u} }$ 
& restriction to the island $u$;  $\mathbb{C}^{A^{cycle}_{0,u}}$ (Sect.~\ref{set:proof})\\

$\zeta: \mathbb{C}^{\tilde{A}_0} \to \mathbb{C}^{\tilde{A}_0\setminus tails}$
& restriction to $\mathbb{C}^{\tilde{A}_0\setminus tails}$ (Sect.~\ref{set:proof})
\end{tabular}
\end{table}

\section{Circulant quantum walk on graphs}
\subsection{Setting of the graph and labeling}
Let $G=(V, A)$ be a connected digraph where $V$ is the set of the vertices and $A$ is the set of arcs. If every arc $a\in A$ has the inverse arc $\bar{a}\in A$, we call this graph a symmetric digraph. The terminal and origin vertices of $a\in A$ are denoted by $t(a)$ and $o(a)$, respectively.
Let $G_0=(V_0, A_0)$ be a finite and connected symmetric digraph.  To this original graph $G_0$, in this paper, we connect the semi-infinite path to {\it every} vertex. This resulting infinite graph is denoted by $\tilde{G}_0=(\tilde{V}_0,\tilde{A}_0)$. 
We set the degree of vertex $u\in V_0$ by \[\deg(u):=\{a\in \tilde{A}_0 \;|\; t(a)=u\}=\{a\in \tilde{A}_0 \;|\; o(a)=u\}.\]
The set of boundary of $G_0$; $\partial A_{\pm}$ are defined by 
\[ \partial A_+=\{t(a)\in V_0,\;o(a)\notin V_0\;\},\;\partial A_-=\{o(a)\in V_0,\;t(a)\notin V_0\}. \]
In the following, we introduce the concept of the labeling, which plays an important role to describe the time evolution of the quantum walks treated here.  Let $A_u\subset \tilde{A}_0$ be the set of arcs whose terminal verteices are commonly $u\in \tilde{V}_0$; that is, $A_u=\{a\in \tilde{A}_0\;|\;t(a)=u\}$.
\begin{definition}\label{def:label}Labeling of arcs {\rm :} 
The labeling of the arcs of  $\tilde{G}_0=(\tilde{V}_0,\tilde{A}_0)$ is defined by the series of the bijection maps $(\xi_u)_{u\in \tilde{V}_0}$. Here $\xi_u$ is a bijection map such that 
\[ A_u \to \{ 0,\dots,\deg(u)-1 \}. \]
\end{definition}
Note that the number of choices of the labeling $\xi=(\xi_u)_u$ of $\tilde{G}_0$ is $\prod_{u\in V_0}\deg(u)!$; we choose a labeling from one of these choices. 

\subsection{Circulant quantum walk}\label{subsec:CQW}
In this subsection, we introduce the circulant quantum walk on the infinite graph $\tilde{G}_0=(\tilde{V}_0,\tilde{A}_0)$ with the labeling $\xi=(\xi_u)_{u\in \tilde{V}_0}$ which has tails defined as described in the previous subsection. 
For a discrete set $\Omega$, we define $\mathbb{C}^\Omega$ as the vector space whose standard basis is described by each element of $\Omega$; that is, $\mathbb{C}^\Omega=\mathrm{span}\{\mathrm{\delta_\omega\;|\; \omega\in \Omega}\}$.  
Here $\delta_\omega$ is 
\[ \delta_\omega(\omega')=\begin{cases} 1 & \text{: $\omega=\omega'$,} \\ 0 &\text{: $\omega\neq \omega '$.} \end{cases} \]. 
For a $2$-dimensional unitary matrix assigned at each vertex $u\in V_0$, 
\[ H_u=\begin{bmatrix}a_u & b_u \\ c_u & d_u \end{bmatrix}, \] we introduce the following $\deg(u)\times \deg(u)$- circulant matrix $\mathrm{Circ}(H_u)$ induced by the $2\times 2$-matrix $H_u$, such that  
$[\mathrm{Circ}(H_u)]_{i,j=0}^{\deg(u)-1}=w^{(u)}_{i-j}$, where $i-j$ is the modulus of $\deg(u)$; that is,  
\begin{equation}\label{eq:circulantmat}
\mathrm{Circ}(H_u)=
    \begin{bmatrix}
    w_0^{(u)} & w_{\deg(u)-1}^{(u)} & w_{\deg(u)-2}^{(u)} &  \cdots & w_1^{(u)} \\
    w_1^{(u)} & w_0^{(u)} & w_{\deg(u)-1}^{(u)} & \cdots & w_2^{(u)} \\
    w_2^{(u)} & w_1^{(u)} & w_0^{(u)} & \cdots & w_3^{(u)} \\
    \vdots    & \ddots & \ddots & \ddots & \vdots \\
    w_{\deg(u)-1}^{(u)} & w_{\deg(u)-2}^{(u)} & w_{\deg(u)-3}^{(u)} & \cdots & w_0^{(u)}
    \end{bmatrix}.
\end{equation}
Here $w^{(u)}_\ell$ is defined by 
\begin{align*}
w_0^{(u)}=d_u+\frac{b_uc_u}{1-a_u^{\kappa}}a_u^{\kappa-1}, \;\; 
    w_{\ell}^{(u)}=\frac{b_uc_u}{1-a_u^\kappa}a_u^{\ell-1}\;\;(\ell=1,\dots,\kappa-1).
\end{align*}
Throughout this paper, we assume $a_ub_uc_ud_u\neq 0$ to avoid a trivial dynamics of the quantum walk.
\begin{assumption}\label{assump:1}
For any $u\in V_0$, we assume
$a_ub_uc_ud_u\neq 0$. 
\end{assumption}
The $\deg(u)\times\deg(u)$-matrix; $\mathrm{Circ}(H_u)$, is a unitary matrix;  
see Lemma~\ref{lem:1} for more detail.  

To explain exactly how the quantum walk iterates the time evolution on the graph $\tilde{G}_0$ with the labeling $\xi$ driven by the circulant matrix, let us introduce $\chi_u: \mathbb{C}^{\tilde{A}}\to \mathbb{C}^{[\deg(u)]}$ by the restriction to $\mathbb{C}^{[\deg(u)]}$ such that 
$(\chi_u\psi)(a)=\psi(\xi_u(a))$. The adjoint operator is described by 
\[ (\chi^*_uf)(a) =\begin{cases} f(\xi^{-1}(a)) & \text{: $a\in A_u$,}\\ 0 & \text{: $a\notin A_u$.} \end{cases} \]
A matrix representation of the map $\chi_u$ is expressed by the $\deg(u)\times \infty$ matrix;
\[ \chi_u\cong [I_{\deg(u)}\;|\; 0\;], \]
under the decomposition of the set of arcs into $A_u \sqcup (\tilde{A}\setminus A_u)$. 
Now we are ready to give the definition of the circulant quantum walk on graph $\tilde{G}$ with the labeling $\xi$. 
\begin{definition}\label{def:CQW} Circulant quantum walk on $\tilde{G}_0$:
$QW(G_0;\{H_u\}_{u\in V_0};\{\xi_u\}_{u\in \tilde{V}_0})$ 
\noindent
\begin{enumerate}
    \item The total vector space: $\mathbb{C}^{\tilde{A}_0}$ 
    \item The time evolution operator: $\tilde{U}_0=U(G_0;\{H_u\}_{u\in V_0};\{\xi_u\}_{u\in \tilde{V}_0})=SC$.
    Here $S$ is the flip flop shift operator such that $(S\psi)(a)=\psi(\bar{a})$ for any $\psi\in \mathbb{C}^{\tilde{A}}$, $a\in \tilde{A}$, and $C$ is defined by 
    \[ C=\bigoplus_{u\in \tilde{V}} \chi_u^* C_u \chi_u\]
    under the decomposition of $\mathbb{C}^{\tilde{A}_0}=\oplus_{u\in \tilde{V}_0}\mathbb{C}^{A_u}$, where 
        \[ C_u=\begin{cases} \mathrm{Circ}(H_u) & \text{: $u\in V_0$} \\
        \sigma_X & \text{: $u\notin V_0$,}
        \end{cases} \]
        and $\sigma_X$ is the Pauri matrix. 
    \item The initial state: 
    \[\psi_0(a)=
    \begin{cases}
    1 & \text{: $a\notin A_0$, $\dist(o(a),V_0)>\dist(t(a),V_0)$,}\\
    0 & \text{: otherwise.}
    \end{cases}\] 
\end{enumerate} 
We call this quantum walk the circulant quantum walk on $G_0$. 
\end{definition}
Let us explain the important points of the time iteration of this quantum walk. 
Let $\psi_n$ be the $n$-th iteration by $\psi_{n+1}=\tilde{U}_0\psi_n$.
 The dynamics on the tails is ``free" such that if the arcs of a tail are labeled by  $a_0,a_1,a_2,\dots$ with $o(a_0)\in V_0$, $t(a_j)=o(a_{j+1})$ ($j=0,1,2,\dots$), then 
\begin{equation}\label{eq:tail}
    \begin{bmatrix}
    \psi_{n+1}(a_{j+1}) \\ \psi_{n+1}(\bar{a}_{j})
    \end{bmatrix}
    =\sigma_X
     \begin{bmatrix}
    \psi_{n}(\bar{a}_{j+1}) \\ \psi_{n}(a_{j})
    \end{bmatrix}
    =\begin{bmatrix}
    \psi_{n}(a_{j}) \\ \psi_{n}(\bar{a}_{j+1})
    \end{bmatrix}.
\end{equation}
This means that a quantum walker is perfectly transmitting at each vertex on the tails. Note that the free quantum walk on the tails is independent of the labeling of the vertices. On the other hand, the quantum walk in the internal graph depends on the labeling. At each vertex on the internal graph $G_0$, a quantum walker is scattered by $\mathrm{Circ}(H_u)$ as follows:   
\begin{equation}
    \begin{bmatrix}
    \psi_{n+1}(\overline{\xi_u^{-1}(0)}) \\
    \vdots \\
    \psi_{n+1}(\overline{\xi_u^{-1}(\deg(u)-1)})
    \end{bmatrix}
    = \mathrm{Circ}(H_u) \begin{bmatrix}
    \psi_{n}(\xi_u^{-1}(0)) \\
    \vdots \\
    \psi_{n}(\xi_u^{-1}(\deg(u)-1))
    \end{bmatrix}
\end{equation}
for any $u\in V_0$.
The initial state $\psi_0$ is set so that $\psi_n(a)=\psi_0(a)$ for any $n\geq 0$ and  $a\in \partial A_+$.  
Therefore a quantum walker is provided to the internal graph $G_0$ as the inflow from $\partial A_+$ while a quantum walker is consumed as the outflow to $\partial A_-$ at every time step.
\section{Optical quantum walk on the blow-up graphs}
In this section, we introduce another quantum walk on the blow-up graph induced by $\tilde{G}_0=(\tilde{V}_0, \tilde{A}_0)$ with the labeling $\xi$. 
This quantum walk is called the optical quantum walk. As we will see, the optical quantum walk is implemented by a circuit of the optical polarizing elements in theory. 
Moreover, the stationary state of the optical quantum walk coincides with that of the circulant quantum walk under some conditions. 

To explain the implementation design and the condition in greater detail, let us first introduce the definitions of the blow-up graph and the optical quantum walk precisely. 
\subsection{blow-up graph}\label{def:BlowupGraph}
Let $\tilde{G}_0=(\tilde{V}_0,\tilde{A}_0)$ be the original graph with the tails. Recall that the bijection map $\xi_u:A_u\to [\deg(u)]:=\{0,\dots,\deg(u)-1\}$ is assigned at each vertex $u$ as defined by the previous section. The labeling is denoted by $\xi:=(\xi_u)_{u\in \tilde{V}_0}$. 
Under this setting, the blow-up graph of $\tilde{G}_0$ is defined as follows. See also Fig.~\ref{fig:GOGBU}. 
\begin{definition}\label{def:BUG}
Blow-up graph of $\tilde{G}_0$ with the labeling $\xi$ {\rm :}  $\tilde{G}_0^{(BU,\xi)}=(\tilde{V}_0^{(BU,\xi)},\tilde{A}_0^{(BU,\xi)})$. \\
The vertex and arc sets are defined as follows. 
\begin{align}
    \tilde{V}_0^{(BU,\xi)} &= V_0^{(BU)}\cup (\tilde{V}_0\setminus V_0), \\
    \tilde{A}_0^{(BU,\xi)} &= A_0^{(BU,\xi)} \cup (\tilde{A}_0\setminus A_0).
\end{align}
Here 
\[V_0^{(BU)}:=\sqcup_{u\in V_0}\{ (u;\xi_u(a)) \;|\;t(a)=u \}\] and 
$A_0^{(BU,\xi)}$ is defined as follows. 
There is an arc from $(u;j)\in V_0^{(BU,\xi)}$ to $(v;\ell)\in V_0^{(BU,\xi)}$ in $\tilde{G}_0^{(BU,\xi)}$; that is, 
$((u;j),(v;\ell))\in A_0^{(BU,\xi)}$, if and only if either of the following two conditions is hold.  
\begin{enumerate}
    \item $u=v$ and $\ell=j+1$ in the modulus of $\deg(u)$, \\
    or
    \item $v=o(\xi_u^{-1}(j))$ and $u=o(\xi^{-1}_v(\ell))$ in $\tilde{G}_0$.
\end{enumerate}
Each tail connecting to vertex $u\in V_0$ in the original graph $\tilde{G}_0$ is connected to the vertex $(u;\xi_u(a))$ with $a\in \partial A_+$ in the blow-up graph $\tilde{G}_0^{(BU,\xi)}$. 
\end{definition}
The blow-up graph is constructed by (1) blowing up each vertex $u\in V(G_0)$ as the {\it directed} cycle following the labeling $\xi_u$, and by (2) connecting to other oriented cycles by the symmetric arcs following the original connection in $G_0$.  
Then if vertices $u$ and $v$ are connected in $G_0$, there are symmetric arcs between the directed cycle of $u$ obtained by (1) and the that of $v$ in $\tilde{G}_0$. We call the directed cycle of $u$ in the new graph $\tilde{G}_0^{(BU,\xi)}$ obtained by (1) the island of $u$.  
The set of all the arcs in all the islands is denoted by $A_0^{cycle}$, which is the set of the arcs of the oriented cycles by the blowing up. 
On the other hand, since the arcs obtained by (2) in Definition~\ref{def:BUG} are isomorphic to $\tilde{A}_0$, we denote this set by $\tilde{A}_0$ itself to reduce the number of notations.  
\begin{remark}
The set of arcs $\tilde{A}_0^{(BU,\xi)}$ is divided into 
    \[ \tilde{A}_0^{(BU,\xi)}=A_0^{cycle} \sqcup \tilde{A}_0. \]
\end{remark}
The blow-up graph is not a symmetric directed graph; that is, the existence of the inverse arc is not ensured, and it has more vertices and arcs than the original graph, which would seem to be suggest a complexity. On the other hand, it has the following nice property. 
\begin{remark}\noindent
The blow-up graph is a $2$-regular digraph; that is,   
for every vertex $u\in \tilde{V}_0^{(BU,\xi)}$, the in-degree and the out-degree are $2$;
one pair of in-arc and out-arc belongs to $A_0^{cycle}$, and the other belongs to $\tilde{A}_0$. 
\end{remark}
The reason for replacing the original vertex $u$ in $\tilde{G}_0$ with island $u$ in $\tilde{G}^{(BU,\xi)}$ is to represent the $2$-polarizations of the optics as follows. 
A quantum optics is driven by ``$2$" polarizations represented by $|V\rangle=[1,0]^\top$ and $|H\rangle=[0,1]^\top$. Instead of such a ``complexity" of the blowing up graph, we obtain a representation of the ``$2$"-internal degrees of freedom on every vertex in the blow-up graph because of the ``$2$"-regularity. The subset $\tilde{A}_0$ keeps the fundamental structure of the original graph. This fact will play an important role in the implementation of the circulant quantum walk on graph $\tilde{G}_0$ by a quantum optics.      
\subsection{Optical quantum walk and the motivation}
The time evolution operator of the optical quantum walk is determined by the parameters of the previous circulant quantum walk. The graph of the optical quantum walk is the blow-up graph $\tilde{G}_0^{(BU,\xi)}$. 
Recall that the in- and out- degrees of the blow-up graph are $2$. 
The vector space of the time evolution is represented by $\mathbb{C}^{\tilde{A}_0^{(BU,\xi)}}$. The scattering at each vertex $(u;j)$ is expressed by $H_u$. More precisely, we define the optical quantum walk as follows. 
\begin{definition}\label{def:OQW}
Optical quantum walk on $\tilde{G}_0^{(BU,\xi)}$: $\mathrm{Opt}(QW(G_0;\bs{H};\xi))$.
\begin{enumerate}
    \item The vector space: $\mathbb{C}^{\tilde{A}_0^{(BU,\xi)}}$.
    \item The time evolution: 
    Let $a_{in}\in A_0^{cycle}$, $b_{in}\in \tilde{A}_0$ be the arcs whose terminal vertices are $(u;j)$, and $a_{out}\in A_0^{cycle}$, $b_{out}\in \tilde{A}_0$ be the arcs whose origins are also $(u;j)$. Then the time evolution operator $U^{(BU)}$ is defined as follows: 
    \begin{equation}\label{eq:OQW} \begin{bmatrix}(U^{(BU)}\psi)(a_{out}) \\ (U^{(BU)}\psi)(b_{out}) \end{bmatrix}=H_u \begin{bmatrix}\psi(a_{in}) \\ \psi(b_{in}) \end{bmatrix} 
    \end{equation}
    for any $\psi\in \mathbb{C}^{\tilde{A}_0^{(BU,\xi)}}$. 
    On the tails, the dynamics of the quantum walk is free; that is, it follows (\ref{eq:tail}). 
    \item The initial state:
    \[\psi_0(a)=
    \begin{cases}
    1 & \text{: $a\notin A_0^{(BU,\xi)}$, $\dist(o(a),V_0^{(BU)})>\dist(t(a),V_0^{(BU)})$,}\\
    0 & \text{: otherwise.}
    \end{cases}\] 
\end{enumerate}
\end{definition}
Our interest is how the optical quantum walk ``imitates" the original circulant quantum walk. One of them can be implemented by a quantum optics in theory. An experimental implementation of $\psi_\infty^{(BU)}$ by optical polarizing elements is proposed in Sect.\ref{sect:ei}.  
According to \cite{HS}, both of the stationary states for the circulant  quantum walk and its induced optical quantum walk exist:
\begin{theorem}[\cite{HS}]\label{thm:stationary}
Let $\psi_n$ and $\psi_n^{(BU)}$ be the $n$-th iterations of the circulant quantum walk and its induced optical quantum walk. Then we have 
\[ \exists \lim_{n\to\infty}\psi_n=:\psi_\infty,\; \exists\lim_{n\to\infty}\psi_n^{(BU)}=:\psi_\infty^{(BU)}. \]
\end{theorem}
We can then focus on their stationary states and the condition for the two  stationary states to coincide. 
\begin{definition}\label{def:implement} Notion of the implementation in this paper {\rm: }
We say that the optical quantum walk implements the underlying circulant quantum walk if 
\[ \psi_\infty^{(BU)}(a)=\psi_\infty(a)\text{ for any $a\in \tilde{A}_0$.} \]
\end{definition}
In Sect.~\ref{sect:demonstration}, we give examples by numerical simulation. 
\section{A circuit of optical polarizing elements for the optical quantum walk }\label{sect:ei}
In this section, we propose the design of the optical circuit implementing the optical quantum walk in an ideal environment where the phase is matched in each interference and there is no attenuation. Improvement points for a more realistic design are discussed in the final section. 

First we introduce our idea for the  implementation of the island in $\tilde{G}_0^{(BU,\xi)}$ by using the half wave plate (HWP) and polarizing beam splitter (PBS) so that the output to arbitrary input of this optical circuit is represented by our  circulant matrix. 
Secondly, we explain how to connect each implemented circuit to reproduce the dynamics on the optical QW. 

\subsection{Design of $U^{(BU)}|_{\text{island} \;u }$}\label{subsect:island}
\begin{figure}
    \centering
    \includegraphics[keepaspectratio, width=150mm]{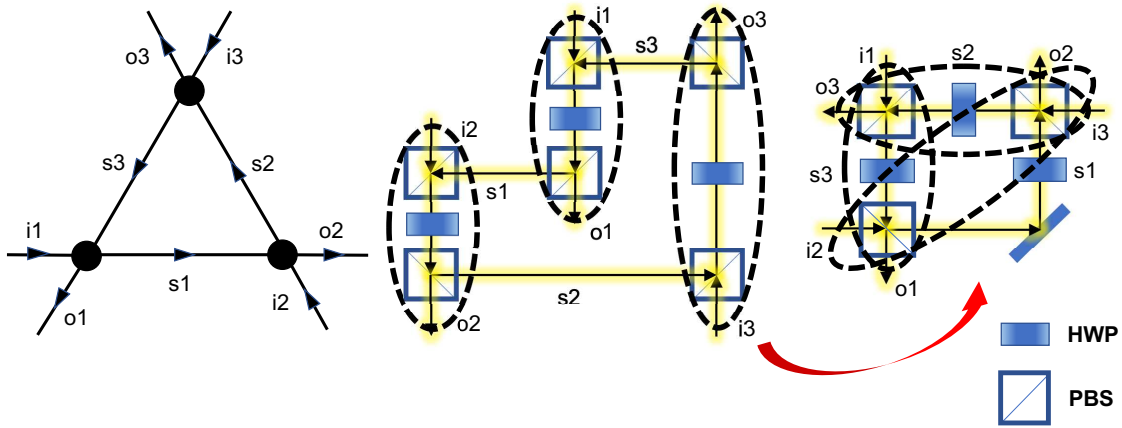}
    \caption{The island of the blow-up graph  (left figure), the island mounted with optical elements (middle figure) and the simplified version of the island mounted with optical elements (right figure): This figure corresponds to the islands of $\tilde{G}^{(BU,\xi)}$ for $N=3$, where each island of $\tilde{G}^{(BU,\xi)}$ is a directed graph with two inputs (i, s) and two outputs (o, s) (for example, in the left figure, we can see the vertices where i1 and s3 are inputs and o1 and s1 are outputs). The island on the right is a simplified equivalent of the middle optical island.}
    \label{fig:a}
\end{figure}
\begin{figure}
    \centering
    \includegraphics[keepaspectratio, width=150mm]{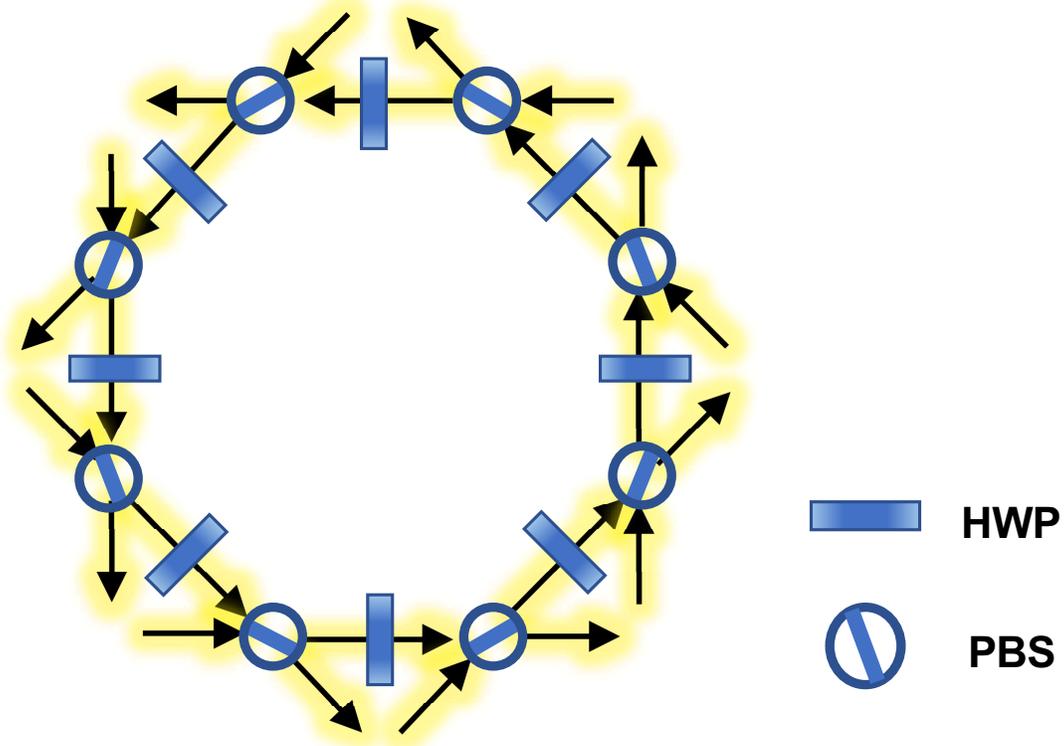}
    \caption{The island implemented with the optical elements shown in the right of Figure 1, with $N$ generalized(the figure shown here is for $N=8$, but it can be extended to general circumscribed polygons): Mirrors are eliminated, and the angle of incidence and reflection at the PBS are not taken into account (originally the angle of incidence and reflection would be 90°, but that is just a matter of changing the direction of the optical path). Each location of HWP corresponds to an individual vertex of the island. Each PBS plays two roles, both receiving the inflows with a superposition of H and V polarizations from the backward HWP and with H polarization from the outside, and sending the outflow with V polarization to the forward HWP and with H polarization to the outside.}
    \label{fig:b}
\end{figure}
In this subsection, we design the ``parts" of the circuit for the optical QW which will be placed on the islands in the blow-up graph. 
The island $u$ is represented by a directed cycle with the tails. Let $\vec{C}_N$ be such a directed graph with $N$ vertices. See Fig.~\ref{fig:a} for the $N=3$ case. 
We introduce an implementation of the quantum walk restricted to the island $\vec{C}_N$ driven by  $2\times 2$ matrix $H$ by optical polarizing elements. Let us explain the implementation by the following three steps. \\

\noindent {\bf Direct implementation ($N=3$).}\\
First, let us focus on the implementation for $N=3$ case.  
The stationary state of the quantum walk on the left figure in Fig.~\ref{fig:a},  $\vec{C}_3$, is described as 
\begin{equation}\label{eq:cycle} \begin{bmatrix} s_{k} \\ o_{k} \end{bmatrix}=H \begin{bmatrix} s_{k-1} \\ i_{k} \end{bmatrix} \end{equation}
for any $k=0,1,2$. 
The middle figure in Fig.~\ref{fig:a} depicts a direct implementation approach in which the dynamics of the quantum walk on the island is mounted with optical elements. Each HWP's in the middle figure is sandwiched between two PBSs (surrounded by dotted lines), which corresponds to one of the vertices of $\vec{C}_3$.    
There are two modes of polarization, H polarization and V polarization. 
The fundamental idea of our implementation is that we establish a correspondence between the arcs of $\vec{C}_3$ and the modes of polarization; that is,  
each arc of the triangle in the left figure corresponds to the V polarization while each arc of the tail in the left figure correspond to the H polarization. 
Since PBS is responsible for transmitting an H polarization and reflecting a V polarization, 
the first PBS in the vertex $k$ ($k=1,2,3$)  represents the situation that the vertex $k$  receives the inflow of quantum walkers from both inside ($s_{k-1}$) and outside ($i_k$) while the second PBS represents the situation that the vertex $k$ sends the outflow to both inside ($s_k$) and outside ($o_k$). We have set the inflow vectors $[s_{k-1},i_{k}]^\top$ and the outflow vectors $[s_k,o_k]^\top$ ($k=1,2,3$) in (\ref{eq:cycle}). 
With PBS alone, there is no factor that can affect the polarization state. Then we set the HWP between the two PBSs, which shifts the phase of each polarization; as a consequence, input H polarization results in a superposition of the H and V polarizations. Thus the sandwiched HWP represents the unitary operator $H$ in (\ref{eq:cycle}). 
In summary, the sandwiched HWP plays the role of the operation of quantum coin $H$, while the first and second PBSs play the role to giving the inflow of a quantum walker to vertex $k$ represented by $[s_{k-1}\;i_k]^\top$ and the outflow of a quantum walker from the vertex $k$ represented by $[s_k\;o_k]^\top$ in (\ref{eq:cycle}), respectively. 

Note that it is possible to make an arbitrary unitary matrix $H$ by setting additional quarter wave plates \cite{polaoptics}.
Also note that there are fixed-end and free-end reflections in $PBS$.
In the situation we are considering here, the V polarization only appears on the $s_k$ corresponding to the vertex $k$ of the island. Therefore, at $i_k$ and $o_k$, where V polarization does not need to be taken into account, there is no reflection at the $PBS$ and the phase $\pi$ is not affected by the fixed-end. For these reasons, the island shown by the right figure in Fig.~1 should be mounted so that the fixed-end faces the outside of the island.\\

\noindent {\bf Economical implementation $(N=3)$}\\
\noindent In the second step, to reduce the optical elements, we introduce an ``economical" design as depicted by the right figure in Fig.~\ref{fig:a}.
Let us explain that we can omit the route between the first PBS in the forward vertex and the second PBS in the backward vertex of the middle figure in Fig.~\ref{fig:a}, which means that we combine the PBS for the inflow placed in vertex $k$ with the PBS for the outflow placed in vertex $k-1$ ($k=0,1,2$). 
In the middle figure of Fig.~\ref{fig:a}, let us focus on the optical route denoted by the arc from the HWP to the second PBS placed in the vertex $3$. Let us denote the state on this arc by $w_3$. This state $w_3$ is described by a superposition of polarizations H and V. The outflow $o_3$ corresponds to the V polarization while $s_3$ corresponds to the H polarization. 
The state $s_3$ will be reflected on the first PBS placed in the vertex $1$. Then the state between the first PBS and HWP placed in the vertex $1$ is described by $s_3+i_1$. In the next, let us focus on the corresponding optical route between the HWP and the combined PBS in the right figure. The state before the combined PBS in the right figure is $w_{3}$ and the state after the combined PBS is split into $s_{3}+i_{1}$ as the inflow to vertex $1$ and $o_3$ as the outflow from vertex $3$. Then the response to the inflows of the left figure is isomorphic to that of the middle figure. Note that the in- and out- flows on each PBS derive from the different vertices in this implementation. \\ 

\noindent {\bf Extension to $N\geq 3$}\\
As a third step, the idea for $N=3$ can be extended to a general $N\geq 3$. 
In Fig.~\ref{fig:b}, we draw the resulting design for general $N$.
The island with the vertices $N$; $\vec{C}_N$, can be considered in the same way as the right figure in Fig.~\ref{fig:a}. 
Each location of the HWP corresponds to each vertex of the island. Each PBS take both roles receiving the inputs from the backward vertex and from the outside, and sending the outputs to the forward vertex and to the outside. Then, the input from the outside corresponds to a quantum walker from the outside to the {\it forward} vertex, while the output to the outside corresponds to a quantum walker to the outside from the {\it backward} vertex. This means that at each PBS, the vertex of the output to the outside shifts that of the input from the outside by one. This observation will be important to design the whole circuit. 

\subsection{Drawing the circuit}\label{subsect:optcircuit}
We build the circuit by combining the above parts (islands). 
We introduce the method drawing the circuit so that each outflow from an island is switched to another inflow to a neighboring island.  
In the following, the blow-up graph $\tilde{G}_0^{(BU,\xi)}$ is deformed to $G'=(HWP\sqcup PBS, A')$ to draw the circuit. After drawing the graph $G'$, we place the half-wave plate on each $HWP$ vertex in the island $u$ and the polarizing beam splitters on each $PBS$ vertex in $G'$. Now let us explain how to draw the graph $G'$ from the blow-up graph $\tilde{G}_0^{(BU,\xi)}$. \\

\noindent{\bf Vertex set}:
We begin by drawing the ``circumscribed polygon" to each original island in $\tilde{G}_0^{(BU,\xi)}$. The vertex set of $G'$ is constructed by all the  corners of the circumscribed polygons, and all the sides of the circumscribed polygons. The vertex subset on the corners is denoted by $PBS$, while the vertex subset on the sides is denoted by $HWP$. The $HWP$ vertex is in one-to-one correspondence with the vertex set of the original graph $\tilde{G}_0^{BU,\xi}$. \\
\noindent{\bf Arc set}:
Following the orientation of each island, the sides of the circumscribed polygon with the subdivision by $HWP$ vertices are replaced with the arcs in $G'$.  
Next, let us define the remaining arcs in $G'$. 
Let us set the arc  
from the islands $u$ to $v$ in $\tilde{G}_0^{(BU,\xi)}$ by $a\in \tilde{A}_0$. 
The origin and terminal vertices of $a$ in $\tilde{G}_0^{(BU,\xi)}$ are denoted by $(u;\xi_u(v))$ and $(v,\xi_v(u))$, respectively.\footnote{Originally, the domain of $\xi_u$ was the arc set whose terminal vertices are $u$ in $\tilde{G}_0$, but since there is a one-to-one correspondence between the above arcs and their origin vertices if $G_0$ is a simple graph, we change the domain to the vertex set for readability. }
In $G'$, the arc $a$ is replaced with the arc whose origin vertex is the $PBS$ vertex located in the corner between the $(u;\xi_u(v))$ and $(u;\xi_u(v)+1)$ and the terminus vertex is the $PBS$ vertex located in the corner between $(v;\xi_v(u))$ and $(v;\xi_v(u)-1)$. The same reconnection procedure is done to every arc in $\tilde{A}_0$. Then the new graph $G'$ is obtained. \\ 

\noindent The reason for the reconnection procedure is as follows. From the consideration of the economical design in subsection~\ref{subsect:island} (see also Fig.~\ref{fig:b}), we know that 
the $PBS$ vertex between the $HWP$ vertex $(u;\xi_u(v))$ and $(u;\xi_u(v)+1)$ must play the two roles, receiving inflow from the islands $\xi_u^{-1}(\xi_u(v)+1)$ and sending outflow to the island $v$. 
It is implies that the forward $PBS$ vertex of the $HWP$ vertex located $(u;\xi_u(v))$ sends the outflow to the island $v$ while the backward $PBS$ vertex of the HWP vertex located in $(u;\xi_u(v)-1)$ receives the inflow from the island $v$. 
By switching the situation of the islands $u$ to $v$, we see that the forward $PBS$ vertex sends the outflow to the island $u$ while the backward vertex receives the inflow from the island $u$. Thus in the optical circuit graph $G'$,
the arc from $(u;\xi_u(v))$ to $(v;\xi_v(u))$ in $\tilde{G}_0^{(BU,\xi)}$ must be replaced with the arc from ``the $PBS$ vertex between the $HWP$ vertices $\xi_u(v)$ and $\xi_u(v)+1$" to ``the $PBS$ vertex between the $HWP$ vertices $\xi_v(v)-1$ and $\xi_v(u)$" in the optical design graph $G'$. Then we realize the situation that each outflow from an island is switched to another inflow to a neighboring island by using the designs of $U^{(BU)}|_{\text{island}\;u}$'s. 
More realistic implementation and arising problems will be discussed in the final section. 
\begin{figure}
    \centering
    \includegraphics[keepaspectratio, width=150mm]{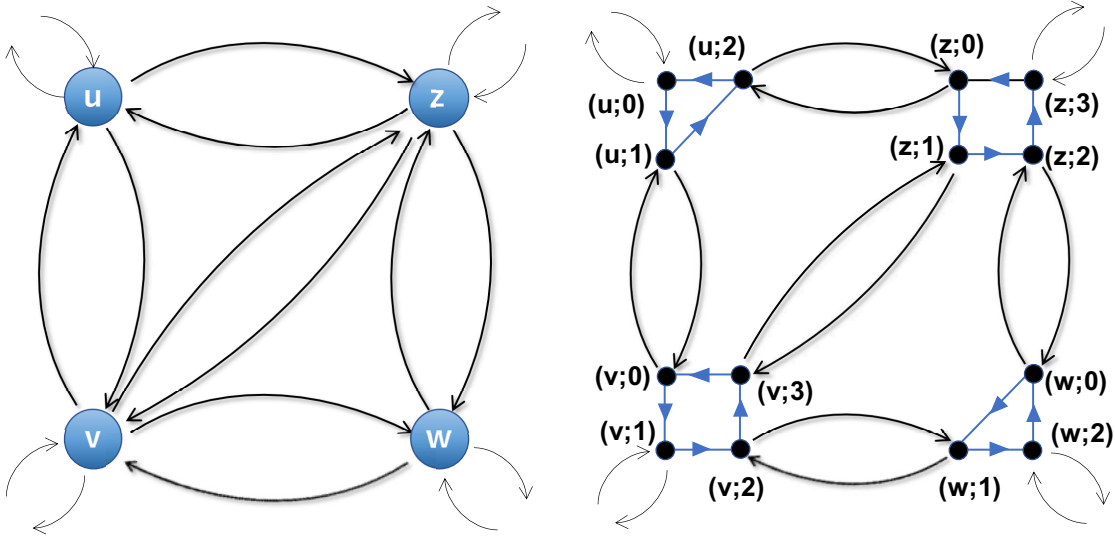}
    \caption{The original graph $G_0$ (left figure) and its blowing up graph $\tilde{G}^{(BU,\xi)}$ (right figure): The labeling $\xi$ is set as follows: $\xi_u(tail)=0$, $\xi_u((v,u))=1$, $\xi_u((z,u))=2$; $\xi_v(tail)=1$, $\xi_v((w,v))=2$, $\xi_v((z,v))=3$, $\xi_v((u,v))=0$; $\xi_w(tail)=2$, $\xi_w((z,w))=0$, $\xi_w((v,w))=1$; $\xi_z(tail)=3$, $\xi_z((u,z))=0$, $\xi_z((v,z))=1$, $\xi_z((w.z))=2$. For example; (i) there is an arc from $(u;2)$ to $(u;0)$ because the terminal vertices are commonly $u$ and $2+1=0$ in the modulus of $\deg(u)=3$, which satisfies the connected condition (1) in Definition~\ref{def:BUG}; (ii) there are symmetric arcs between vertex $(v;2)$ and $(w;1)$ in $\tilde{G}_0^{(BU;\xi)}$, because we can check that $o(\xi_v^{-1}(2))=o((w,v))=w$ and $o(\xi_w^{-1}(1))=o((v,w)))=v$, which satisfies the connected condition (2) in Definition~\ref{def:BUG}.}
    \label{fig:GOGBU}
\end{figure}
\begin{figure}
    \centering
    \includegraphics[keepaspectratio, width=150mm]{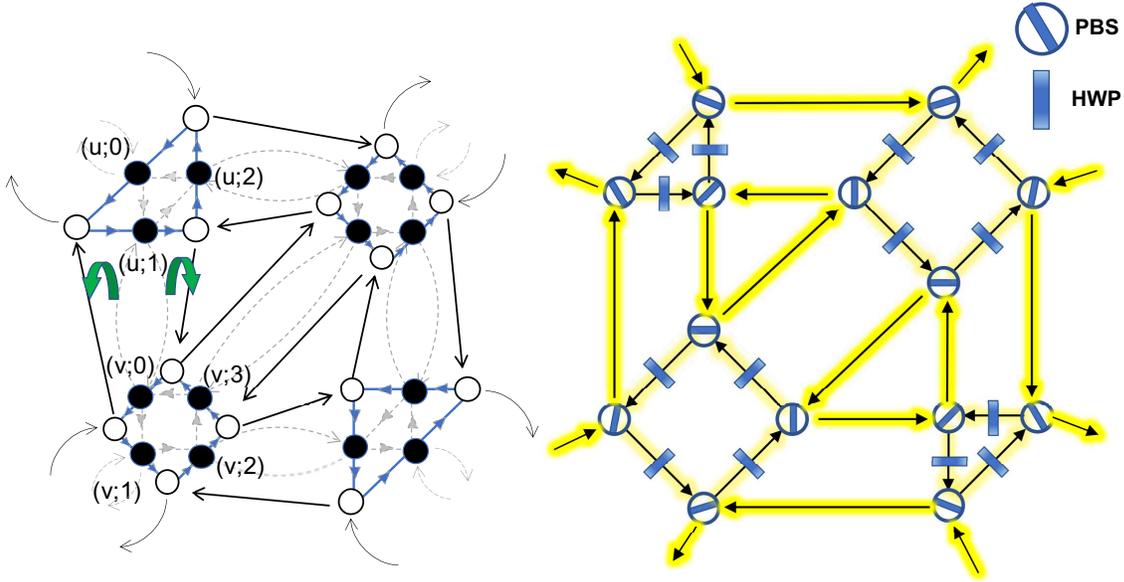}
    \caption{$G'$ (left figure) and the design of the optical circuit (right figure): To highlight the act of drawing the circuit, we include  $\tilde{G}^{(BU,\xi)}$ in gray color in the left figure. We draw the ``circumscribed polygons" around each island (blue arcs) and place the vertices on the corners and in the midpoint of each side. Then, the sides of the circumscribed polygon correspond to  $A_0^{cycle}$. On the other hand, all the arcs in $\tilde{A}_0$ are reconnected as depicted in the left figure. For example, the symmetric arcs between $(u;1)$ and $(v;0)$ are reconnected by changing the origin and terminal vertices to the corresponding corners of the circumscribed polygon. We set the half wave plate on each black vertex and the polarizing beam splitter on each white vertex as depicted in the right figure. }
    \label{fig:GprimeO}
\end{figure}
\section{Demonstration by numerical simulations}\label{sect:demonstration}
In this section, we examine whether a circulant quantum walk on the complete graph with $10$ vertices, $K_{10}$, is implemented by the corresponding optical quantum walk by changing ${\bs H}$.  
The arc whose terminus is $i$ and origin is $j$ is denoted by $(i,j)$, and the arc whose terminus is $i$ and origin is located in the tail is denoted by $(i,i)$ in $\tilde{G}_0$. The labeling  $\xi=(\xi_0,\dots,\xi_{9})$ is given as follows:  
\begin{align}\label{eq:label}
    \xi_i((i,j)) &= j 
\end{align}
for every $i,j=0,\dots,9$. 
We consider following two examples. First, we consider the case in which the the implementation is realized, and then we consider the case that the implementation fails. 
For this purpose, we define a relative probability of the circulant QW by 
    \[ \mu_n(j)=\sum_{a\in \tilde{A}_0:t(a)=j} |\psi_n(a)|^2. \]
We also set
    \[\mu_n^{BU}(j)=\sum_{a\in \tilde{A}_0:t(a)=j} |\psi_n^{BU}(a)|^2,\]
and chase their time courses $\mu_n$ and $\mu_n^{BU}$ simultaneously.
Note that the summation is taken over $\tilde{A}_0$ in the definition of $\mu_n^{(BU)}$ because the ``implementation" is determined at the original arcs of $G_0$ by Definition~\ref{def:implement}. 
The existence of limits of $n\to\infty$ to $\mu_n$ and $\mu_n^{BU}$ is ensured by Theorem~\ref{thm:stationary}. We call  $\lim_{n\to\infty}\mu_n$ a stationary measure of the circulant QW. \\

\noindent {\bf Example 1.} {\it $K_{10}$ with a marked vertex} :
We choose the vertex $0$ in the set of vertices $\{0,\dots,9\}$ as the marked vertex.   
The circulant coin assigned at each vertex is denoted by $\mathrm{Circ}(H_j)$. 
Here $H_j$ is set so that a perturbation is given only at the marked element $0$ as follows. 
\[ H_j=
\begin{cases} 
\begin{bmatrix}
1/\sqrt{2} & 1/\sqrt{2} \\ 1/\sqrt{2} & -1/\sqrt{2} 
\end{bmatrix} & \text{: $j=0$} \\
\\
\begin{bmatrix}
1/\sqrt{2} & 1/\sqrt{2} \\ -1/\sqrt{2} & 1/\sqrt{2} 
\end{bmatrix} & \text{: otherwise}
\end{cases} \]
for $j=0,\dots,9$.
Figure~\ref{Fig:1} shows the time course of the relative probability at each vertex $0,\dots,9$. The blue curve describes the time course of the relative probability of the vertex $0$, which is the marked vertex. We observe that although the time scales of the convergence are different, the stationary measures converge to the same value for every vertex. This is theoretically supported by Corollary~\ref{cor:graph} ,because there are arcs in $G_0^{(BU;\xi)}$ satisfying condition (2) in the corollary; these arcs are $(0,1)$ and $(0,9)$ and their inverses. 
\begin{figure}
    \begin{tabular}{cc}
      \begin{minipage}[t]{0.45\hsize}
        \centering
        \includegraphics[keepaspectratio, width=70mm]{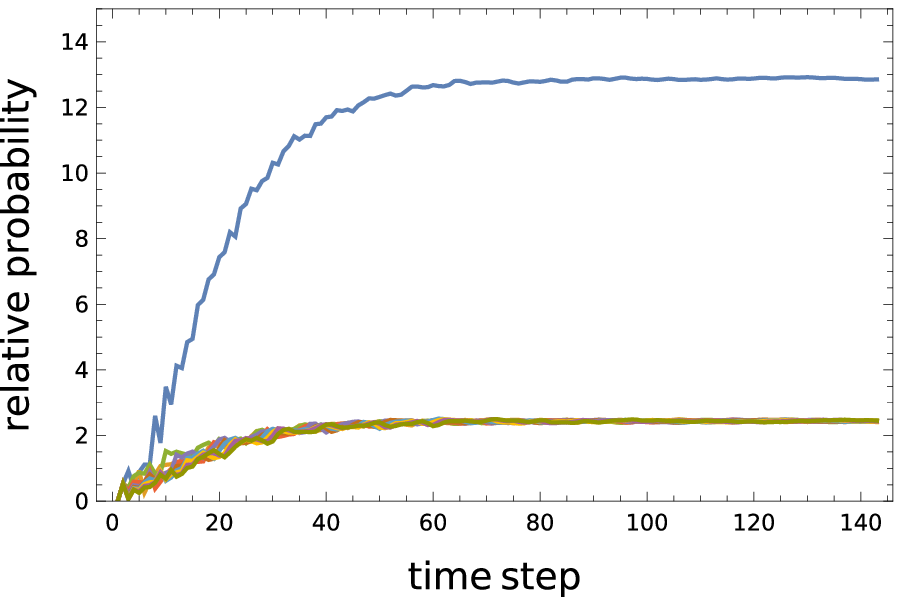}
        \caption{Time course of $\mu_n$ for the circulant QW with a marked vertex in the setting of Example~1: the horizontal and vertical lines are the time step and the relative probability of each vertex. The blue line is the time course of the relative probability of the marked vertex, the other lines are the time courses of the other vertices.  }
        \label{Fig:1}
      \end{minipage} &
      \begin{minipage}[t]{0.45\hsize}
        \centering
        \includegraphics[keepaspectratio, width=70mm]{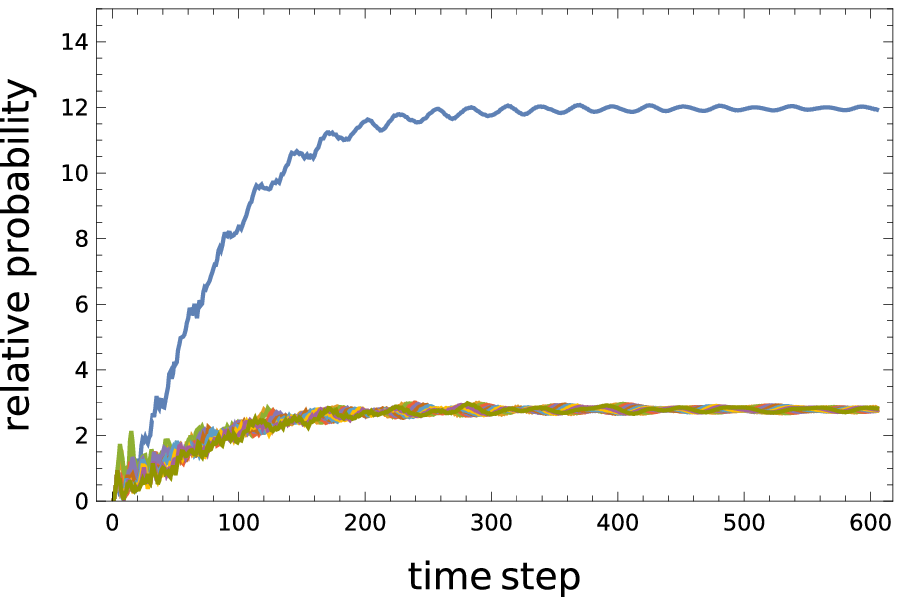}
        \caption{Time course of $\mu_n^{BU}$ for the optical QW with a marked vertex induced by the left circulant QW}
        \label{Fig:2}
      \end{minipage}
    \end{tabular}
  \end{figure}

\noindent {\bf Example 2.}($K_{10}$ with the uniform setting) We assign $H_j$'s uniformly by
\[ H_j=\begin{bmatrix} 1/\sqrt{2} & 1/\sqrt{2} \\ 1/\sqrt{2} & -1/\sqrt{2} \end{bmatrix} \]
for any $j=0,\dots,9$. 
From the symmetry on the time evolution with respect to each vertex, the stationary measure on each vertex is the same.  
We observe that although both the circulant QW and its induced optical QW converge to some stationary measures, the convergence values are different. 
\begin{figure}
    \begin{tabular}{cc}
      \begin{minipage}[t]{0.45\hsize}
        \centering
        \includegraphics[keepaspectratio, width=70mm]{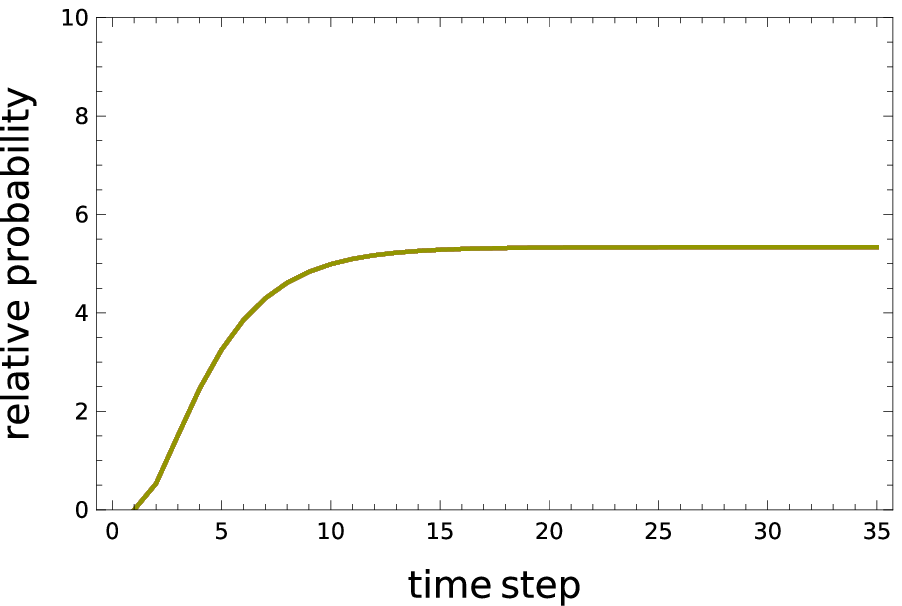}
        \caption{Time course of the relative probabilities of the uniform circulant QW in the setting of Example~2.}
        \label{Fig:1'}
      \end{minipage} &
      \begin{minipage}[t]{0.45\hsize}
        \centering
        \includegraphics[keepaspectratio, width=70mm]{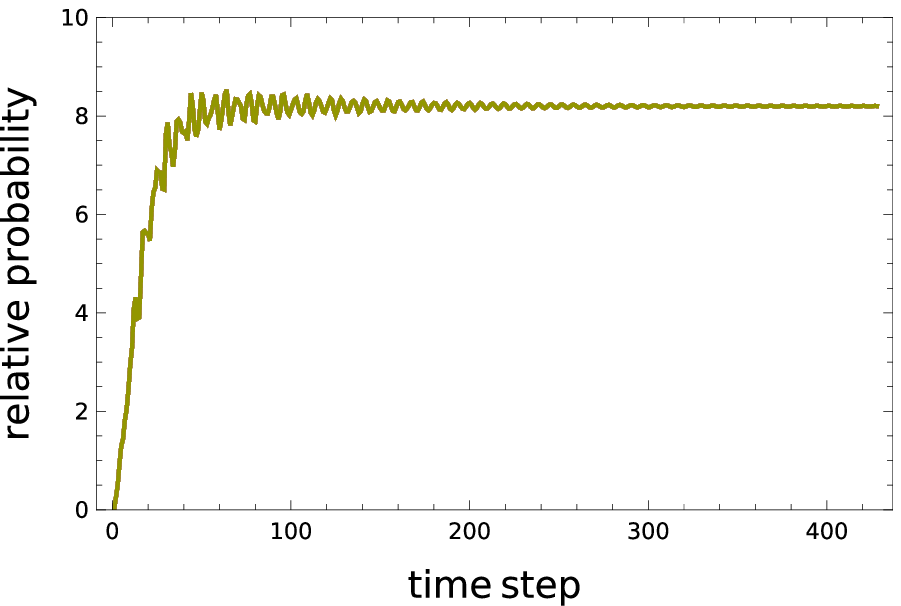}
        \caption{Time course of the relative probabilities of the uniform optical QW induced by the left circulant QW}
        \label{Fig:2'}
      \end{minipage}
    \end{tabular}
  \end{figure}

\section{Mathematical results}
The following sufficient condition for the accomplishment of the implementation is the key to our main results.  
\begin{proposition}\label{thm:key}
Let $\bs{H}=(H_u)_{u\in V_0}$ and $\xi=(\xi_u)_{u\in V}$ in the setting of the circulant quantum walk $QW(G_0;\bs{H};\xi)$ on $\tilde{G}_0$. 
If $\ker(1-U^{(BU)})=\{\bs{0}\}$, 
then  
$\mathrm{Opt}(QW(G_0;\bs{H};\xi ))$ on $\tilde{G}_0^{(BU,\xi)}$ implements $QW(G_0;\bs{H};\xi)$ on $\tilde{G}_0$ .  
\end{proposition}
\begin{proof}
See Section~\ref{set:proof}.
\end{proof} 
The following theorem gives a sufficient condition for $\ker(1-U^{(BU)})= \{\bs{0}\}$,  which illustrates two kinds of the setting of the optical quantum walk. See Figures~\ref{fig:setting1} and \ref{fig:setting2}; a kind of symmetry is broken around the boundary because (1) the rotational orientations of the connected islands are opposite each other or (2) the assigned coins of the connected islands are different. 
\begin{theorem}\label{cor:graph} {\rm (The symmetry breaking designs for the implementation)}\\
Assume that 
\[H_u=\begin{bmatrix} a_u & b_u \\ c_u & d_u \end{bmatrix}\] satisfies $a_ub_uc_ud_u\neq 0$ for any $u\in V_0$. If there is an arc $e\in A_0$ satisfying each of them:  
\begin{enumerate}
    \item $\xi_{o(e)}(\bar{e})-\xi_{o(e)}(t)=\xi_{t(e)}(e)-\xi_{t(e)}(\tau)=\pm 1$,  or 
    \item $\xi_{o(e)}(\bar{e})-\xi_{o(e)}(t)=-(\xi_{t(e)}(e)-\xi_{t(e)}(\tau))=\pm 1$ and $d_{o(e)}\neq d_{t(e)}^*$, 
\end{enumerate} 
where $t,\tau\in \partial A_+$ whose terminal vertices are located in the islands $o(e)$ and $t(e)$, respectively,  
then $\mathrm{Opt}(QW(G_0;\bs{H};\xi ))$ implements $QW(G_0;\bs{H};\xi)$.
Here for a complex number $z$, $z^*$ is the conjugate of $z$. 
\end{theorem}
\begin{figure}
    \begin{tabular}{cc}
      \begin{minipage}[t]{0.45\hsize}
        \centering
        \includegraphics[keepaspectratio, width=75mm]{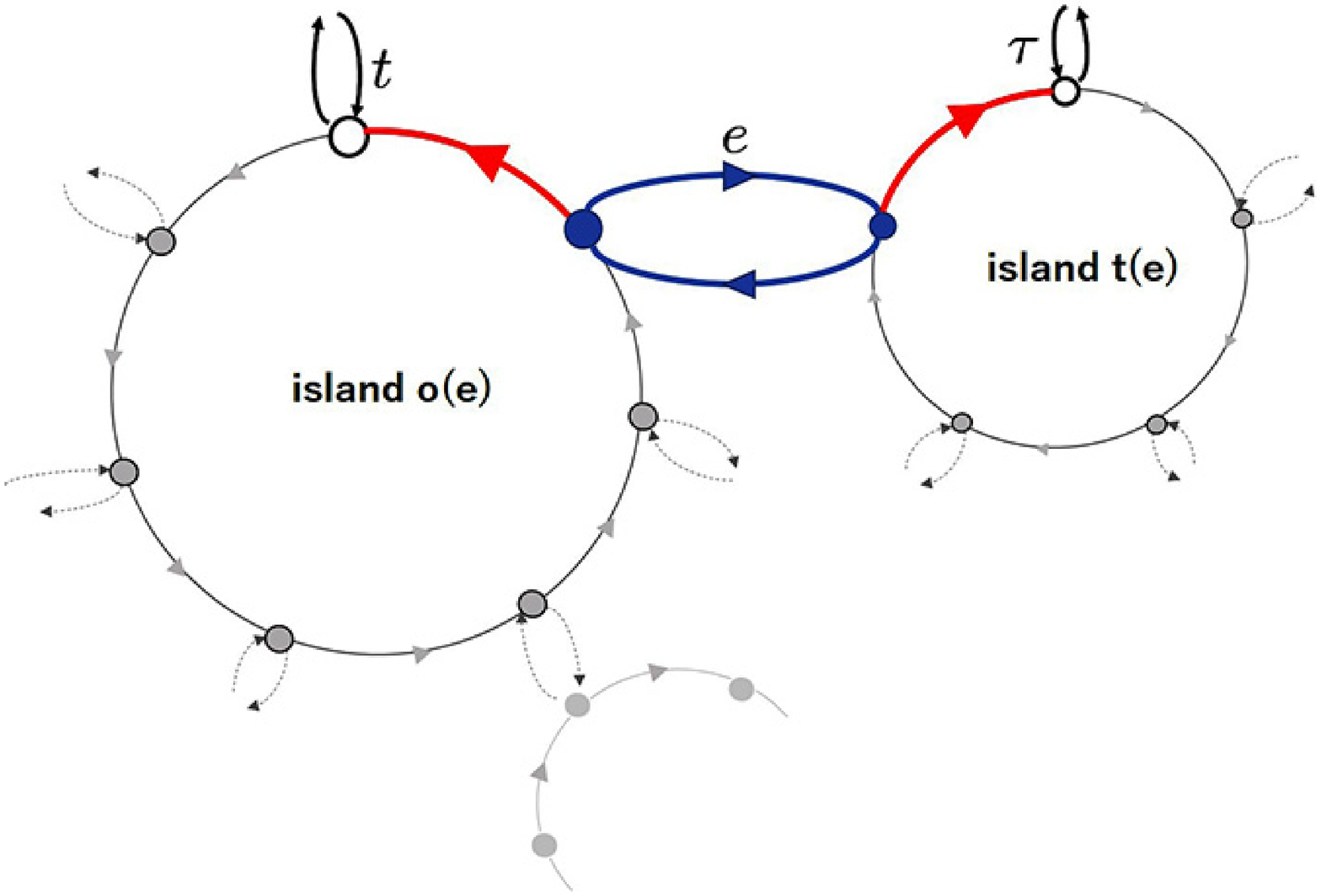}
        \caption{The symmetry breaking design in Theorem~\ref{cor:graph} (1) for the implementation }
        \label{fig:setting1}
      \end{minipage} &
      \begin{minipage}[t]{0.45\hsize}
        \centering
        \includegraphics[keepaspectratio, width=75mm]{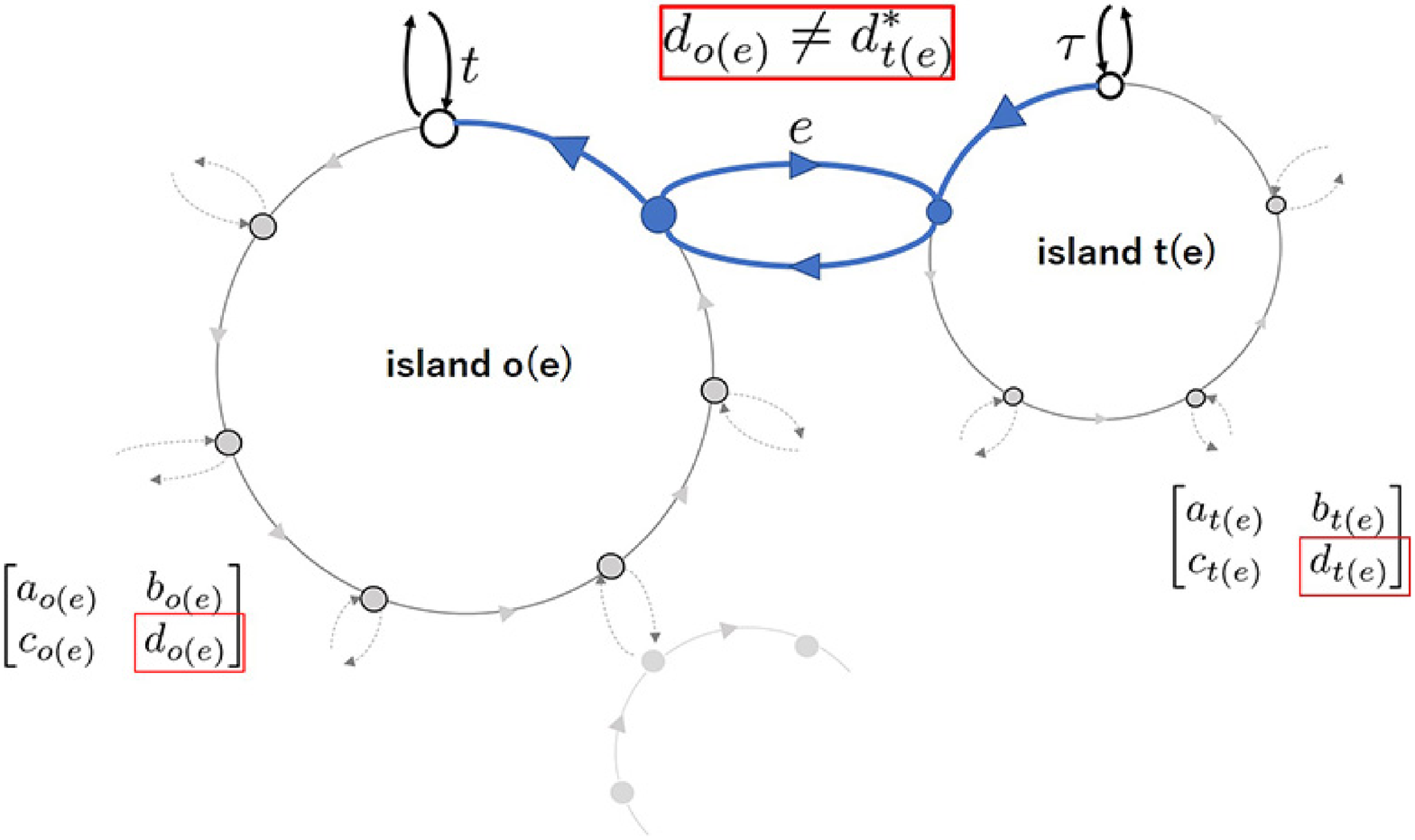}
        \caption{The symmetry breaking design in Theorem~\ref{cor:graph} (2) for the implementation}
        \label{fig:setting2}
      \end{minipage}
    \end{tabular}
  \end{figure}
Example~1 in Section~\ref{sect:demonstration} matches the setting of case (2) in Corollary~\ref{cor:graph}. This is the reason that the stationary state of the circulant QW coincides with that of its optical QW. 
Now we give the proof using Proposition~\ref{thm:key}. 
\begin{proof}
The arc in $\partial A_+$ whose terminal vertex is $u\in V_0$ is denoted by $\tau_u$.
We set $B_u\subset A^{(BU,\xi)}_0$ as the set of arcs whose terminal or origin vertices are $(u,\xi_u(\tau_u))$ (see Fig.~\ref{fig:Thm6.1}):  
\[ B_u= \{a\in A^{(BU,\xi)}_0 \;|\; o(a)=(u,\xi_u(\tau_u))\text{\;or\;}t(a)=(u,\xi_u(\tau_u))\}. \]

The eigenvector of eigenvalue $1$ is denoted by $\psi$. Note that $\psi$ satisfies  not only $U^{BU,\xi}\psi =\psi$ but also $||\psi||^2<\infty$. 
First, let us see the following fact holds: 
\begin{equation}\label{eq:condition1}
    \supp(\psi) \cap (\cup_{u\in V_0}B_u) =\emptyset. 
\end{equation}
In the tail, $\psi(\tau_u)=\psi(a_1)=\psi(a_2)=\cdots$ and $\psi(\bar{\tau}_u)=\psi(\bar{a}_1)=\psi(\bar{a}_2)=\cdots$ (where $o(\tau_u)=t(a_1),o(a_1)=t(a_2),\dots$) holds.
Then to ensure $||\psi||<\infty$, the values $\psi(\tau_u)$ and $\psi(\bar{\tau}_u)$ must be $0$. 
From the definition of the time evolution operator of this quantum walk, letting $e_{in}',e_{out}'\in B_u$, we have 
    \[ \begin{bmatrix} \psi(e_{out}') \\ \psi(\bar{\tau}_u) \end{bmatrix}=\begin{bmatrix} a_u & b_u \\ c_u & d_u \end{bmatrix}\begin{bmatrix} \psi(e_{in}') \\ \psi(\bar{\tau}_u) \end{bmatrix}, \]
which is equivalent to  
\[ \psi(\tau_u)=(\psi(e_{out}')-a_u\psi(e_{in}'))b_u^{-1},\;\psi(\bar{\tau}_u)=c_u\psi(e_{in}')+(\psi(e_{out}')-a_u\psi(e_{in}'))b_u^{-1}d_u. \]
This implies that $\psi(\tau_u),\psi(\bar{\tau}_u)= 0$ if and only if  $\psi(e_{out}'),\psi(e_{in}')=0$. Then (\ref{eq:condition1}) holds. 

Considering the contraposition of Proposition~\ref{thm:key}, we notice that it is enough to show that if $\ker(1-U^{BU})\neq\{\bs{0}\}$, then 
there are no arcs in $A_0$ satisfying (1) or (2). 
For $a\in A_0$, let us put $e_1,e_2\in A^{cycle}$ with $t(e_1)=(o(a);\xi_{o(a)}(\bar{a}))$ and $t(e_2)=(t(a);\xi_{t(a)}(a))$ (see Figure~\ref{fig:Thm6.1}).
Assume (1) holds, and consider the $+1$ case. By (\ref{eq:condition1}), we have
\[ \begin{bmatrix} 0 \\ \psi(a) \end{bmatrix} = \begin{bmatrix} a_{o(a)} & b_{o(a)} \\ c_{o(a)} & d_{o(a)}\end{bmatrix}\begin{bmatrix} \psi(e_1) \\ \psi(\bar{a}) \end{bmatrix}\text{ and } \begin{bmatrix} 0 \\ \psi(\bar{a}) \end{bmatrix} = \begin{bmatrix} a_{t(a)} & b_{t(a)} \\ c_{t(a)} & d_{t(a)}\end{bmatrix}\begin{bmatrix} \psi(e_2) \\ \psi(a) \end{bmatrix}.\]
From the unitarity of $H_{o(a)}$ and $H_{t(a)}$, we have $\psi(a),\psi(\bar{a})\neq 0$. Taking the inverses of $H_{o(a)}$ and $H_{t(a)}$ to both equations, which are the adjoints of them, we also have $\psi(\bar{a})=d_{o(a)}^*\psi(a)$ and $\psi(a)=d_{t(a)}^*\psi(\bar{a})$. 
This implies $|d_{o(a)}d_{t(a)}|=1$. The unitarity of $H_u$'s leads to $|d_{o(a)}|=|d_{t(a)}|=1$, which induces $b_u=c_u=0$ $(u\in \{o(a),t(a)\})$. This is the contradiction to the Assumption~\ref{assump:1}. In the same way, the case for $-1$ in (1) can be proved. 
In the next, let us assume (2) holds. Consider $+1$ case. Then we have 
\[ \begin{bmatrix} 0 \\ \psi(a) \end{bmatrix} = \begin{bmatrix} a_{o(a)} & b_{o(a)} \\ c_{o(a)} & d_{o(a)}\end{bmatrix}\begin{bmatrix} \psi(e_1) \\ \psi(\bar{a}) \end{bmatrix}\text{ and } \begin{bmatrix} \psi(\bar{e}_2) \\ \psi(\bar{a}) \end{bmatrix} = \begin{bmatrix} a_{t(a)} & b_{t(a)} \\ c_{t(a)} & d_{t(a)}\end{bmatrix}\begin{bmatrix} 0 \\ \psi(a) \end{bmatrix}.\]
Taking the inverse, which is the adjoint of $H_{o(a)}$, to the first equation, we have $d_{o(a)}^*\psi(a)=\psi(\bar{a})$, while computing the second equation directly, we have $\psi(\bar{a})=d_{t(a)}\psi(a)$; which implies $\psi(\bar{a})/\psi(a)=d^*_{o(a)}=d_{t(a)}$. 
The $-1$ case also can be done in the same way. 
\end{proof}
\begin{figure}
    \centering
    \includegraphics[keepaspectratio, width=100mm]{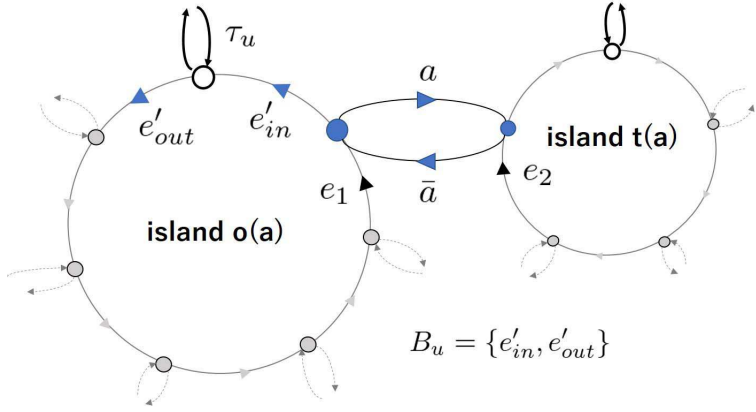}
    \caption{The vertex colored white is the vertex connecting to a tail. The eigenvector of eigenvalue $1$ does not overlap to $B_u$, for any island $u$. }
    \label{fig:Thm6.1}
\end{figure}
\section{The conditions of $\ker(1-U^{BU})=0$ for $K_N$ case}
In this section, let us consider the optical quantum walk induced by the circulant QW with the uniform circulant coin on the complete graph with $N$ vertices $K_N$; $QW(K_N;\bs{H};\xi)$. Here the labeling $\xi=(\xi_i)_{i=0}^{N-1}$ is the same as in (\ref{eq:label}) and the unitary matrices  $\bs{H}=(H_i)_{i=0}^{N-1}$ are denoted by 
    \[ H_i=H=\begin{bmatrix}a& b \\ c & d \end{bmatrix} \]
with $abcd\neq 0$. 
In this setting, we obtain a useful sufficient condition to the parameters $a,b,c,d$ of the circulant QW for the implementation as follows.  
\begin{theorem}
Let the circulant QW on $K_N$ be set as described in the above.  
If ``$d\neq \mathbb{R}$" or ``$\det(H)^{2N}\neq 1$ for $(N>3)$,\;$\det(H)^{2N}\neq -1$ for $(N=3)$", then the induced optical QW implements the circulant QW.
\end{theorem}
In the setting of Example~2, the above condition is not satisfied because $\det(H)=-1$ ($a=b=c=-d=1/\sqrt{2}$). Now let us move to the proof. 
\begin{proof}
Let us first prepare the following key lemma for the proof. 
\begin{lemma}\label{lem:Kn}
Let $N\geq 3$. We have 
\[\dim\ker(1-U^{BU}) = 
\begin{cases}
N/2-1 & \text{: $d\in \mathbb{R}$, $(-\Delta)^{2N}=1$, $N$ is even,} \\
(N-3)/2 & \text{: $d\in \mathbb{R}$, $(-\Delta)^{2N}=1$, $(-\Delta)^{N}\neq 1$, $N$ is  odd,} \\
(N-1)/2 & \text{: $d\in \mathbb{R}$,  $(-\Delta)^{N}= 1$, $N$ is odd,} \\
0 & \text{: otherwise.}
\end{cases}\]
\end{lemma}
By Lemma~\ref{lem:Kn}, if ``$d\notin \mathbb{R}$" or ``$\Delta^{2N}\neq 1$", Theorem~\ref{thm:key} implies that the optical QW implements the underlying circulant QW.  
\end{proof}
Now we focus on the proof of Lemma~\ref{lem:Kn}.\\
\noindent{\bf Proof of Lemma~\ref{lem:Kn}}. 
Assume $\ker(1-U^{BU})=0$. The labeling $\xi$ satisfies with the condition (2) in Corollary~\ref{cor:graph}. Then the $(2,2)$-element of $H$; $d$, must be a real number.  
Let $U^{BU}\Psi=\Psi$ and $\supp(\Psi)$ be included in the internal blow-up graph. Each vertex in $\tilde{G}_0^{BU,^\xi}$ is labeled by $(\ell,m)$. Here $\ell$ represents the island and $m$ represents the island heading from the island $\ell$. The arc whose origin is $(\ell,m)$ and terminus is $(\ell,m+1)$ belongs to $A_0^{cycle}$, where $m+1$ is the modulus of $N$, while the arc whose origin and terminus are $(\ell,m)$, $(m,\ell)$, respectively belongs to $\tilde{A}_0$. We define $(\ell,\ell)$ is the vertex connecting to the tail. We also define $((\ell,\ell),(\ell,\ell))$ as the arc from $(\ell,\ell)$ to the tail. 
We set 
    \[ \Psi((\ell,m-1),(\ell,m))=:z_m^{(\ell)}, \text{ and } \Psi((\ell,m),(m,\ell))=:x_{\ell,m}.  \]
Note that since the support of $\Psi$ is included in the internal graph, $x_{\ell,\ell}=0$ for any $\ell=0,\dots,N-1$.
All the indexes of ``$z$" and ``$x$" are the modulus of $N$.
Then we have the following useful lemma.
\begin{lemma}
Put $\Delta:=\det H$. We have 
\begin{align}
\begin{bmatrix}
x_{\ell,m} \\ x_{m,\ell}
\end{bmatrix}
&= \frac{1}{b} \begin{bmatrix}
d & 1 \\ 1 & d
\end{bmatrix} 
\begin{bmatrix}
z_{m+1}^{(\ell)} \\ -\Delta z_{m}^{(\ell)}
\end{bmatrix},\label{eq:1} \\ 
\begin{bmatrix}z_{m+1}^{(\ell)} \\ -\Delta z_m^{(\ell)}\end{bmatrix}  &= \sigma_X   \begin{bmatrix}z_{\ell+1}^{(m)} \\ -\Delta z_\ell^{(m)}\end{bmatrix} \label{eq:z} 
\end{align} 
for any $\ell,m=0,\dots,N-1$. 
\end{lemma}

\begin{proof}
By inserting the time evolution of the optical QW in Definition~\ref{def:OQW} into the eigenequation $U^{BU}\Psi=\Psi$, we have
\begin{align}
    \begin{bmatrix}z_{m+1}^{(\ell)}\\x_{\ell,m}\end{bmatrix}=\begin{bmatrix}a & b \\ c & d\end{bmatrix}\begin{bmatrix}z_m^{(\ell)} \\ x_{m,\ell}\end{bmatrix},\;\;
    \begin{bmatrix}z_{\ell+1}^{(m)}\\x_{m,\ell}\end{bmatrix}=\begin{bmatrix}a & b \\ c & d\end{bmatrix}\begin{bmatrix}z_\ell^{(m)} \\ x_{\ell,m}\end{bmatrix}.
\end{align}
Solving $x_{\ell,m}$ and $x_{m,\ell}$ from each equation, we obtain
\begin{align}\label{eq:xz}
    \begin{bmatrix} x_{\ell,m} \\ x_{m,\ell} \end{bmatrix} = \frac{1}{b} \begin{bmatrix} d & -\Delta \\ 1 & -\Delta d \end{bmatrix} \begin{bmatrix} z_{m+1}^{(\ell)} \\ z_{m}^{(\ell)} \end{bmatrix},\;
    \begin{bmatrix} x_{m,\ell} \\ x_{\ell,m} \end{bmatrix} = \frac{1}{b} \begin{bmatrix} d & -\Delta \\ 1 & -\Delta d \end{bmatrix} \begin{bmatrix} z_{\ell+1}^{(m)} \\ z_{\ell}^{(m)} \end{bmatrix}.
\end{align}
Here we used $a=\Delta \bar{d}=\Delta d$ which derives from the unitarity of $H$ and Corollary~\ref{cor:graph}. 
Then putting $\bs{z}_m^{(\ell)}:=[z_{m+1}^{(\ell)},\;z_{m}^{(\ell)}]^\top$, we have 
\begin{align*}
    \frac{1}{b} \begin{bmatrix} d & -\Delta \\ 1 & -\Delta d \end{bmatrix} \bs{z}_{m}^{(\ell)} = \sigma_X \frac{1}{b} \begin{bmatrix} d & -\Delta \\ 1 & -\Delta d \end{bmatrix} \bs{z}_{\ell}^{(m)} 
    & \Leftrightarrow 
     \begin{bmatrix} d & 1 \\ 1 & d \end{bmatrix} \begin{bmatrix} 1 & 0 \\ 0 & -\Delta \end{bmatrix}\bs{z}_{m}^{(\ell)} = \sigma_X \begin{bmatrix} d & 1 \\ 1 & d \end{bmatrix} \begin{bmatrix} 1 & 0 \\ 0 & -\Delta \end{bmatrix} \bs{z}_{m}^{(\ell)} \\
    & \Leftrightarrow 
     \begin{bmatrix} 1 & 0 \\ 0 & -\Delta \end{bmatrix}\bs{z}_{m}^{(\ell)} = \begin{bmatrix} d & 1 \\ 1 & d \end{bmatrix}^{-1} \sigma_X \begin{bmatrix} d & 1 \\ 1 & d \end{bmatrix} \begin{bmatrix} 1 & 0 \\ 0 & -\Delta \end{bmatrix} \bs{z}_{\ell}^{(m)} \\
    & \Leftrightarrow 
    \begin{bmatrix} 1 & 0 \\ 0 & -\Delta \end{bmatrix}\bs{z}_{m}^{(\ell)} = \sigma_X  \begin{bmatrix} 1 & 0 \\ 0 & -\Delta \end{bmatrix} \bs{z}_{\ell}^{(m)}.
\end{align*}
Here we used 
    \[ \begin{bmatrix}1 & d \\ d & 1\end{bmatrix}\sigma_X= \sigma_X\begin{bmatrix}1 & d \\ d & 1\end{bmatrix} \]
in the last equivalence. 
\end{proof}
From this Lemma, we obtain two observations. \\
\noindent{\bf Observation 1.}\\
Inserting $\ell=m$ in (\ref{eq:xz}), we have $z_{\ell+1}^{(\ell)}=z_{\ell}^{(\ell)}=0$ since $x_{\ell,\ell}=0$. This is consistent with Corollary~\ref{cor:graph} in the case of (2). \\
\noindent{\bf Observation 2.}\\
Since $\ell,m$ are arbitrarily chosen from $0,\dots,N-1$, equation (\ref{eq:z}) is equivalent to $z_{m}^{(\ell)}=(-\Delta) z_\ell^{(m-1)}$ for any $\ell,m=0,\dots,N-1$. 
Starting from this, we recursively obtain 
\begin{align}
    z_m^{(\ell)} &= (-\Delta)z_{\ell}^{(m-1)} = (-\Delta)^2 z_{m-1}^{(\ell-1)}  \notag\\
    &= (-\Delta)^3 z_{\ell-1}^{(m-2)} = (-\Delta)^4 z_{m-2}^{(\ell-2)} \notag\\
    &= \cdots  \notag\\
    &= (-\Delta)^{2N-1}z_{\ell-N+1}^{(m-N)} = (-\Delta)^{2N}z_{m-N}^{(\ell-N)} = (-\Delta)^{2N}z_{m}^{(\ell)}. \label{eq:orbit}
\end{align}
Then we have $(-\Delta)^{2N}=1$ or $z_{m}^{(\ell)}=0$. If $(-\Delta)^{2N}=0$, then $z_m^{(\ell)}=0$ for all $\ell,m$ and (\ref{eq:z}) implies $x_{\ell,m}=0$ for all $\ell,m$. This is contradiction. Thus at least, $(\Delta)^{2N}=1$ must be satisfied.

In the set of pairs of superscript and subscript $[i,j]\in\{ [\ell-k,m-k]\;|\;k=0,\dots,N \}$ of ``$z$" in the third column of the above equations, the pair $[\ell,m]$ first appears again as $[m-N,\ell-N]$; on the other hand, in the second column, there might be a $k<N$ such that $[\ell-k,m-k-1]=[m,\ell]$ in the modulus of $N$. But such a $k$ can be identified with $k=\ell-m=(N-1)/2$. Thus if such an appearance of $k$ happens, the size $N$ must be odd and the arc of $z_{m}^{(\ell)}$ is located in the ``front" of the vertex $(\ell,\ell)$. In particular, if $N=3$, since $z_{\ell+1}^{(\ell)}=z_\ell^{(\ell)}$, by the observation 1, every arc in the support of $\Psi$ living in $A_0^{cycle}$ is located in the front of the vertex $(\ell,\ell)$ ($\ell=0,1,2$). Therefore $(-\Delta)^{2k+1}=(-\Delta)^{N}$ must be $1$ for $N=3$. If $N>3$ and $N$ is even, the length of such a cycle is $2N$ from (\ref{eq:orbit}). This means that we obtain the eigenvector which is constructed by the orbit of $2N$-length closed path starting from the arc of $z_m^{(\ell)}$ and returning back to this arc. The number of arcs in $A_0^{(cycle)}$ is $N^2$, but Observation~1 implies the $2N$ arcs are eliminated as the support of $\Psi$. Then we have $(N^2-2N)/(2N)=N/2-1$ linearly independent eigenvectors of eigenvalue $1$; that is, $\dim\ker(1-U^{BU})=N/2-1$ if $N>3$ is even, and $(\Delta)^{2N}=1$. 
On the other hand, if $N>3$ and $\ell-m\neq (N-2)/2$, then the length of the orbit of the closed path represented by (\ref{eq:orbit}) is $2N$. Then if $N>3$ is odd and $(-\Delta)^{2N}=1$ but $(-\Delta)^{N}\neq 1$, we have $(N^2-(2N+N))/(2N)$ linearly independent eigenvectors of eigenvalue $1$; while $N>3$ is odd and $(-\Delta)^N=1$, we have $(N^2-(2N+N))/(2N)+1$ linearly independent eigenvectors of eigenvalue $1$. We have completed the proof of Lemma~\ref{lem:Kn}.\;$\square$ 

\section{Proof of key statement (Proposition~\ref{thm:key})}\label{set:proof}
Let us consider the stationary state of the optical quantum walk $\psi^{BU}$ that satisfies $U^{BU}\psi^{BU}=\psi^{BU}$. 
The arcs of the island $u$ in $\tilde{G}_0^{BU,\xi}$ is denoted by  
\[A^{cycle}_{0,u}:=\{a\in A^{cycle}_0 \;|\; t(a)=u\text{ in $\tilde{G}_0$}\}=\{ f_0,\dots,f_{\deg_{\tilde{G}_0}(u)-1} \}\subset \tilde{A}_0^{(BU,\xi)}.\] 
Here we set $t(f_j)=(u;j)$ in $\tilde{G}_0^{BU,\xi}$ for $j=0,\dots,\deg_{\tilde{G}_0}(u)$ (see Fig.~\ref{fig:Lem8.1}). 
The restriction to the island $u$ is defined by $\eta_u: \mathbb{C}^{\tilde{A}^{(BU,\xi)}_0}\to \mathbb{C}^{A^{cycle}_{0,u} }$ such that $(\eta_u\psi)(a)=\psi(a)$ for any $a\in A^{cycle}_{0,u}$ and $\psi\in \mathbb{C}^{\tilde{A}_0^{(BU,\xi)}}$; the adjoint is 
\[ (\eta_u^*f)(a) = \begin{cases} f(a) & \text{: $a\in A_{0,u}^{cycle}$,}\\ 0 & \text{: otherwise.}\end{cases} \]
Note that $\eta_u\eta_u^{*}$ is the identity operator on $\mathbb{C}^{A^{cycle}_0}$ while $\eta^*_u\eta_u$ is the projection operator. 
Then we have 
\begin{align*}
    U^{BU}\psi^{BU}=\psi^{BU} & \Rightarrow \eta_u U^{BU}\psi^{BU}=\eta_u \psi^{BU} \\
    & \Leftrightarrow \eta_uU^{BU}(\eta_u^*\eta_u+(1-\eta_u^*\eta_u))\psi^{BU}=\eta_u\psi^{BU} \\
    & \Leftrightarrow E_u\varphi_{u,0}+\rho_u = \varphi_{u,0},
\end{align*}
where $E_u:=\eta_u U^{BU}\eta_u^*$, $\varphi_{u,0}:=\eta_u\psi^{BU}$, and $\rho_u:= \eta_u U^{BU}(1-\eta_u^*\eta_u)\psi^{BU}$. 
Let $P_u$ be the cyclic permutation matrix on $\mathbb{C}^{[\kappa]}$ such that $(P_u\phi)(j)=\phi(j+1)$ in the modulus of $\kappa=\deg_{\tilde{G}_0}(u)$. Then it is easy to see that $E_u$ is isomorphic to $a_uP_u$. Since $|a_u|<1$, the inverse matrix $(1-E_u)$ exists. Then we have 
\begin{equation}
    \varphi_{u,0}=(1-E_u)^{-1}\rho_{u}, 
\end{equation}
where if we label the $\kappa:=\deg_{\tilde{G}_0}(u)$ arcs from the outside of the island $u$ whose terminal vertices in $\tilde{G}^{BU,\xi}_0$ are $(u;0),(u;1),\dots,(u;\kappa-1)$ by $e_0,\dots,e_{\kappa-1}$, respectively, (see Fig.~\ref{fig:Lem8.1})  then the inflow penetrating into $A_u^{cycle}$ is represented by  
\[\rho_u=[\;b_u\;\psi^{BU}(e_0),\dots,b_u\;\psi^{BU}(e_{\kappa-1})\;]^\top.\] 
Recall that the arcs in the island $u$ are denoted by  $f_0,\dots,f_{\kappa-1}$ and the arcs from the outside of the island $u$ are denoted by $e_0,\dots,e_{\kappa-1}$.
Then we have the following lemma. 
\begin{lemma}\label{lem:1}
Let $\psi^{BU}$ be a generalized eigenvector of $U^{BU}$ satisfying $U^{BU}\psi^{BU}=\psi^{BU}$. 
Set $\varphi_{u,in}=[\;\psi^{BU}(e_0),\dots,\psi^{BU}(e_{\kappa-1})\;]^\top$, $\varphi_{u,out}=[\;\psi^{BU}(\bar{e}_0),\dots,\psi^{BU}(\bar{e}_{\kappa-1})\;]^\top$, and $\varphi_{u,0}=[\; \psi^{BU}(f_0),\dots,\psi^{BU}(f_{\kappa-1}) \;]^\top$. Then we have 
\begin{align}
\varphi_{u,out}&=\mathrm{Circ}(H_u) \varphi_{u,in}, \label{eq:w1}\\
\varphi_{u,0}&=\frac{1}{c_u}(\mathrm{Circ}(H_u)-d_uI_u)\varphi_{u,in}. \label{eq:naka0}
\end{align} 
\end{lemma}
\begin{proof}
It holds that
\begin{align}
    \varphi_{u,out}(j) &= c_u\varphi_{u,0}(j)+d_u\varphi_{u,in}(j) \label{eq:naka}\\
    &= c_u ((1-E_u)^{-1}\rho_u)(j)+d_u \varphi_{u,in}(j) \notag \\
    &= c_u ((1-E_u)^{-1} b_u\varphi_{u,in})(j)+d_u \varphi_{u,in}(j).\notag
\end{align}
Then we have 
\[ \varphi_{u,out}=(b_uc_u(1-E_u)^{-1}+d_u)\varphi_{u,in}. \]
The inverse of $(1-E_u)$ can be expressed by 
\begin{align*}
    (1-E_u)^{-1} &= 1+a_uP_u+(a_uP_u)^2+\cdots \\
    &= 1+(a_u+a_u^{\kappa+1}+a_u^{2\kappa+1})P_u+(a^2+a^{\kappa+2}+a^{2\kappa+2}+\cdots)P_u^2+\\
    &\quad\cdots+(a^{\kappa-1}+a^{\kappa+\kappa-1}+a^{2\kappa+\kappa-1}+\cdots)P_u^{\kappa-1} \\
    &= 1+\frac{a_u}{1-a_u^{\kappa}} P_u+\frac{a_u^2}{1-a_u^{\kappa}} P_u^2+\cdots+\frac{a_u^{\kappa-1}}{1-a_u^\kappa}P_u^{\kappa-1}. \end{align*}
Then inserting the above expression for $(1-E)^{-1}$ into $b_uc_u(1-E_u)^{-1}$, we obtain $\mathrm{Circ}(H_u)=b_uc_u(1-E_u)^{-1}+d_uI_u$, which completes the proof of (\ref{eq:w1}).  We obtain (\ref{eq:naka0}) by combining (\ref{eq:naka}) with (\ref{eq:w1}). 
\end{proof}
\begin{figure}[h]
    \centering
    \includegraphics[width=60mm]{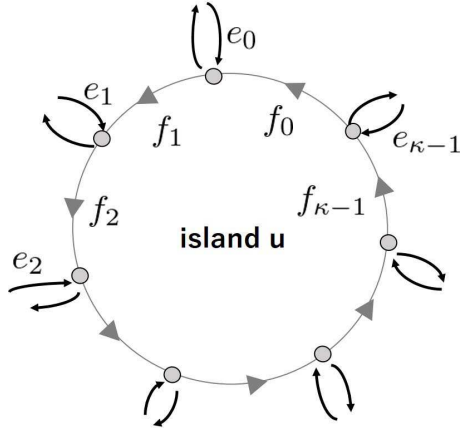}
    \caption{Labeling of the arcs of the island $u$ in the proof of Lemma~\ref{lem:1}. }
    \label{fig:Lem8.1}
\end{figure}

Thus $\mathrm{Circ}(H_u)$ represents the local scattering matrix of the $n$-directed cycle with $n$ boundaries for the in- and out-flow. The unitarity of $\mathrm{Circ}(H_u)$ is ensured by \cite{FelHil1,FelHil2} in more general settings. 

Next, let us introduce the restriction $\iota: \mathbb{C}^{\tilde{A}_0^{(BU,\xi)}}\to \mathbb{C}^{\tilde{A}_0}$ such that $(\iota \psi)(a)=\psi(a)$ for any $a\in \tilde{A}_0$. 
A matrix representation of $\iota$ is given by
\[ \iota \cong [\; I_{\tilde{A}_0} | 0\;] \]
under the decomposition of $\tilde{A}_0^{(BU,\xi)}$ into $\tilde{A}_0\sqcup A^{cycle}_0$. 

Then we have the following proposition.
\begin{proposition}\label{prop:1}
For any $\psi\in \mathbb{C}^{\tilde{A}_0^{(BU,\xi)}}$ satisfying 
a genralized eigenequation $U^{BU}\psi=\psi$, we have 
\[U_0\iota \psi =\iota \psi.\]
\end{proposition}
\begin{proof}
By (\ref{eq:w1}) in Lemma~\ref{lem:1}, 
we immediately obtain the conclusion. 
\end{proof}
We expect that $\iota \psi$ is the stationary state of the circulant quantum walk, but unfortunately, it is not true in general. See Example~2 in Section~\ref{sect:demonstration}. So in the following, let us consider when $\iota \psi$ coincides with the stationary state of the circulant quantum walk. To this end, we prepare the following two propositions. 
\begin{proposition}\label{prop:2}
$\ker(1-U_0)\neq \{\bs{0}\}$ if and only if $\ker(1-U^{BU})\neq \{\bs{0}\}$.
\end{proposition}
\begin{proof}
Assume $\ker(1-U_0)\neq \{\bs{0}\}$, that is, there exists $\phi \neq \bs{0}$ such that $U_0\phi=\phi$ and $\supp(\phi)\subset A_0$. 
Let $\{e_0,\dots,e_{\kappa-1}\}=\{a\in \tilde{A}_0 \;|\;t(a)=u\}$ and  
$\varphi_{u,in}=[\phi(e_0),\dots,\phi(e_{\kappa-1})]^\top$.
Then let us define $\psi\in \mathbb{C}^{\tilde{A}_0^{(BU,\xi)}}$ by 
    \begin{equation}\label{eq:extension}
    \psi(a)=\begin{cases} \phi(a) & \text{: $a\in \tilde{A}_0$,} \\
    \frac{1}{c_u}((\mathrm{Circ}(H_u)-d_uI_u)\varphi_{u,in})(a) & \text{: $a\in A^{cycle}_{0,u}$, $u\in V_0$.}\end{cases} 
    \end{equation}
Then by Lemma~\ref{lem:1}, we have $U^{BU}\psi=\psi$, and the support of $\psi$ does not have the overlap with the tails, which implies that $||\psi||<\infty$. 
Next, let us consider the converse direction. 
Note that $\mathrm{spec}(E_u)\subset \{z\in\mathbb{C} \;|\; |z|<1\}$. Then the support of the eigenvectors of eigenvalue $1$ must has the overlap to $\tilde{A}_0$.  
Then from Proposition~\ref{prop:1}, we obtain the converse direction. 
\end{proof}
\begin{remark}
Proposition~\ref{prop:2} implies that the eigenvectors of $\ker(1-U_0)$ have one-to-one correspondence to those of $\ker(1-U^{BU})$.  
\end{remark}
\begin{proposition}\label{prop:3}
Assume $\ker(1-U_0)=\{\bs{0}\}$. Then for arbitrary $s_0\in \mathbb{C}^{\partial A_+}$, the following generalized eigenfunction satisfying the boundary condition $\psi|_{\partial A_+}=s_0$ is uniquely determined. 
\begin{align*}
(1-U_0)\psi &= 0,\;
\psi|_{\partial A_+}=s_0.
\end{align*}
\end{proposition}
\begin{proof}
Let $\zeta: \mathbb{C}^{\tilde{A}_0} \to \mathbb{C}^{\tilde{A}_0\setminus tails}$ such that $(\zeta \phi)(a)=\phi(a)$ for any $a\in \mathbb{C}^{\tilde{A}_0\setminus tails}$. 
Putting $\zeta U_0 \zeta^*=E$, we have 
\[ (1-E)\zeta \psi =  \rho, \]
where $\rho=\zeta U_0 (1-\zeta^*\zeta)\psi$.
Let us see that $(1-E)^{-1}$ exists as follows. 
Assume there is a $f\neq 0$ such that $Ef=f$.
Taking the operation $\zeta^*$ to both sides, we have $||\zeta^*\zeta U_0\zeta^*f||^2=||\zeta^*f||^2$. 
On the other hand, by the unitarity of $U_0$, we have 
$||U_0\zeta^*f||^2=||\zeta^*f||^2$. 
Then we have $|| U_0\zeta^*f ||=||\zeta^*\zeta U_0\zeta^*f||$. Since $\zeta^*\zeta$ is a projection onto the internal graph, this implies that the support of $U_0\zeta^*f$ is included in the internal graph; that is,  $(1-\zeta^*\zeta)U_0\zeta^*f=0$.
This is equivalent to $U_0\zeta^*f-\zeta^*Ef=0$. Since $Ef=f$, we have $(1-U_0)\zeta^*f=0$. 
Thus $\zeta^*f(\neq 0)$ is an $(+1)$-eigenvector of $U_0$, which contradicts the the assumption $\ker(1-U_0)=\{\bs{0}\}$. Therefore we have $\zeta \psi=(1-E)^{-1}\rho$, which implies that $\zeta^*\zeta \psi=\xi (1-E)^{-1}\rho$. Thus the restriction of $\psi$ to the internal graph is uniquely determined. On the other hand, the restriction of $\psi$ to the tails is uniquely determined as $s_0$ for the inflow and $(1-\zeta^*\zeta)U\zeta^*\zeta\psi$ for the outflow.  
\end{proof}
Now we are ready for the proof of Proposition~\ref{thm:key}; ``if $\ker(1-U^{BU})=\{\bs{0}\}$, then $\mathrm{Opt}(QW(G_0;\bs{H};\xi))$ implements $QW(G_0;\bs{H};\xi)$. "\\
\noindent \\
\noindent{\bf Proof of Proposition~\ref{thm:key}.} 
Let $\Psi$ be the stationary state of the optical quantum walk. 
We set the boundary condition by $\Psi|_{\partial A_+}=s_0$. 
The assumption of Proposition~\ref{thm:key} is equivalent to $\ker(1-U_0)=\{\bs{0}\}$ by Proposition~\ref{prop:2}, which is the assumption of Proposition~\ref{prop:3}. 
Then by Proposition~\ref{prop:3}, the unique solution of $(1-U_0)\psi=0$ with $\psi|_{\partial A_+}=s_0$ is nothing but the stationary state of the circulant quantum walk with the inflow represented by $s_0$. 
On the other hand, 
by Proposition~\ref{prop:1}, we have $U_0\iota\Psi=\iota\Psi$. In particular, $\iota \Psi|_{\partial A_+}=\Psi|_{\partial A_+}=s_0$. 
Therefore, the unique solution $\psi$ is described by $\iota \Psi$, which implies that the stationary state for the circulant quantum walk is equivalent to $\psi=\iota \Psi$. \begin{flushright}
$\square$
\end{flushright}
\begin{remark}
In general, the stationary state of quantum walks must be orthogonal to every eigenspace of the time evolution operators in the whole space~\cite{HS}. Let $\psi_\infty^{(BU)}$ be the stationary state of the optical QW. Then 
for any $\phi\in \ker(1-U^{(BU,\xi)})$, we have 
\[ \langle \psi_\infty^{(BU)}, \phi \rangle=0. \]
However, in general, it is not ensured that 
\begin{equation}\label{eq:orthogonal}
\langle \iota\psi_\infty^{(BU)}, \iota\phi \rangle=0. 
\end{equation} 
This is the reason that the optical QW does not implement the underlying QW in Example~2 since the optical QW has eigenvalue $1$ and does not satisfy (\ref{eq:orthogonal}).
\end{remark}

\begin{remark}
Conversely, if $\ker(1-U_0)=\{\bs{0}\}$, then the stationary state $\psi_\infty$ implements the stationary state of $\psi_\infty^{(BU)}$ by the manner of deformation of $\psi_\infty$ given by (\ref{eq:extension}) in  Proposition~\ref{prop:2}.
This means that two stationary states on $\tilde{G}_0$ and $\tilde{G}^{(BU,\xi)}$ implement each other under the condition of $\ker(1-U_0)=\{\bs{0}\}$. 
\end{remark}
\section{Summary and discussion}
In this paper, we considered an optical implementation of quantum walk on a graph driven by a circulant matrix, namely the circulant quantum walk. To this end, we introduced another kind of quantum walk on the blow-up graph of the original graph induced by the circulant quantum walk; this was the optical quantum walk. The blow-up graph is a $2$-regular directed graph. Then making a correspondence between the two incoming edges to each vertex and the vertical and horizontal polarizations, 
we heuristically showed that the optical quantum walk can be implemented by the optical circuit in theory and also proposed the design of the optical circuit for the general graph.  We suggest a kind of search of a perturbed vertex using our circulant quantum walk. A high relative probability of the perturbed vertex is asymptotically stable, while such a probability is asymptotically periodic in the usual quantum search algorithm driven by quantum walks. 
The analysis on this convergence speed has a potential to considering a quantum walk version of the cut off phenamina~\cite{Diaconis,LevinPeres} in the future. 
We also mathematically showed a sufficient condition for the coincidence of the stationary states of the circulant quantum walk and its induced optical quantum walk. From this condition, we gave a useful setting for the circulant quantum walk that can be implemented by the induced optical quantum walk.  

Finally, let us discuss the design of the optical circuit in Section~4, as an experimental approach to potential problem in the future. In Section 4, we designed the optical circuit under ideal conditions where the phases are matched, but experimentally the phases are not matched due to the noise arising from complex environmental fluctuation, so the expected operation does not occur in HWPs. To solve this problem, we propose an experimental  method.
In an optical circuit such as that in subsection 4.1, it is necessary to make the optical path length of one round trip of V-polarized light be an integer multiple of the wavelength in order to obtain constructive interference with the light from previous laps. It is also necessary to match the phases of waves at each PBS. To meet these requirements, we stabilize the optical path length between each pair of PBS's. Experimentally, stabilization of the length of the optical circuit can be achieved by a feedback control by using a reference laser and a piezo-electric transducer which can be attached to a mirror consisting of the optical circuit~\cite{quantumoptics}. If necessary, the phase of the incoming signal is also stabilized by a similar procedure. Also, in subsection 4.2, each pair of island is connected together following the original graph connection; the resulting design is described by $G'$. The ideal design $G'$ is implemented with optical elements as shown in the right figure of Fig. 4. The polarized light that flows out from $G'$ does not return to $G'$, so there is no need to consider the optical path length. On the other hand, H-polarized light flowing out from one island to another interferes with V-polarized light at the PBS of the destination island, so the optical path length needs to be  stabilized by the same feedback control. However, in the above method, we should measure the outflow from each $PBS$ in order to perform feedback control in the optical path between each $PBS$. Then, a trade off problem remains in that the more accurately we try to get the interference inside the island, the more we lose the output outside the island. We expect that such realistic experimental problems based on our proposed optical circuit under the very ideal condition will be improved in the future.

\noindent\\
\noindent {\bf Acknowledgments}
Yu.H. acknowledges financial supports from the Grant-in-Aid of
Scientific Research (C) Japan Society for the Promotion of Science (Grant No.~18K03401). 
E.S. acknowledges financial supports from the Grant-in-Aid of
Scientific Research (C) Japan Society for the Promotion of Science (Grant No.~19K03616) and Research Origin for Dressed Photon.



\begin{small}
\bibliographystyle{jplain}

\end{small}

\end{document}